\documentclass[11pt,a4paper]{article}

\usepackage[top=3cm, bottom=4cm, left=3cm, right=3cm]{geometry}

\usepackage{amsmath,amssymb}
\usepackage{amsthm}
\usepackage{accents}
\usepackage[active]{srcltx}
\usepackage{xspace}
\usepackage{cite}

\newcommand{\emi}{({\em i}\,)\xspace}
\newcommand{\emii}{({\em ii}\,)\xspace}
\newcommand{\emiii}{({\em iii}\,)\xspace}

\newcommand{\N}{\mathbb{N}}

\newcommand{\R}{\mathbb{R}}

\newcommand{\la}{\lambda}
\newcommand{\eps}{\epsilon}

\newcommand{\Ec}{{\cal E}}
\newcommand{\Mc}{{\cal M}}
\newcommand{\xib}{\boldsymbol{\xi}}

\newcommand{\Star}[1]{*'\!_{#1}\,}
\newcommand{\Starxi}[1]{*^\xi_{#1}\,}
\newcommand{\Starth}[1]{*^\theta_{#1}\,}

\newcommand{\ovcirc}[1]{\accentset{\circ}{#1}}

\newcommand{\we}{\wedge}
\newcommand{\ot}{\otimes}

\newcommand{\vth}{{\vec{\theta}}}

\DeclareMathOperator{\hd}{\ast}
\DeclareMathOperator{\tg}{{\rm tg}}
\DeclareMathOperator{\Exp}{{\rm Exp}}
\DeclareMathOperator{\arctg}{{\rm arc\,tg}}

\DeclareMathOperator{\sgn}{{\rm sgn}}

\DeclareMathOperator{\lr}{\!\lrcorner}

\newtheorem{thr}{Theorem}

\newtheorem{lm}[thr]{Lemma}

\numberwithin{equation}{section}
\numberwithin{thr}{section}

\begin{document}

\title{Constraints of the Teleparallel Equivalent of General Relativity in a gauge}
\author{ Andrzej Oko{\l}\'ow}
\date{May 22, 2024}

\maketitle
\begin{center}
{\it  Institute of Theoretical Physics, Warsaw University,\\ ul. Pasteura 5, 02-093 Warsaw, Poland\smallskip\\
oko@fuw.edu.pl}
\end{center}
\medskip

\begin{abstract}
We consider a specific Hamiltonian formulation of the Teleparallel Equivalent of General Relativity, where the canonical variables are expressed by means of differential forms. We show that some ``position'' variables of this formulation can be always gauge-transformed to zero. In this gauge the constraints of the theory become simpler, and the other ``position'' variables acquire a nice geometric interpretation that allows for an alternative, clearer form of the constraints. Based on these results we derive some exact solutions to the constraints. 
\end{abstract}

\section{Introduction}

General Relativity in a Hamiltonian form is a constrained system \cite{adm}---canonical variables on a spatial slice of a spacetime must satisfy some constraints in order to serve as initial values for Hamilton's equations defining the time evolution of the variables. This statement holds true also in the case of various versions of the Teleparallel Equivalent of General Relativity (TEGR)---see \cite{nester,wall-av,maluf-2,bl,maluf-1,maluf,mal-rev,oko-tegr-I,oko-tegr-II,q-suit,ham-nv,ferraro,ham-tegr-rev} for descriptions of Hamiltonian formulations of TEGR and \cite{tegr-cosm-rev} for a recent extensive review of the theory. 

For some formulations of TEGR the constraints are rather complicated functions of canonical variables. This complexity can make it difficult to analyze the constraints and is perhaps the reason, why it is difficult to find in the literature exact solutions to TEGR constraints. The goal of the present paper is to introduce a gauge, which makes the constraints of a specific version of TEGR much simpler. Moreover, in this gauge it will be natural to change the way the variables are described, which will result in an alternative, more transparent form of the simplified constraints. This will then allow to derive some exact solutions to the constraints. 

More precisely, we will consider a rather non-standard Hamiltonian formulation of TEGR introduced in \cite{q-suit,ham-nv}. In this formulation the ``position'' variables are functions $(\xi^I)$, ($I\in\{1,2,3\}$), and one-forms $(\theta^I)$. The momentum $\zeta_I$ conjugate to $\xi^I$ is a three-form, while the momentum $r_I$ conjugate to $\theta^I$ is a two-form. All the variables are defined on a three-dimensional spatial slice $\Sigma$ of a spacetime and $(\theta^I)$ form a global coframe on the slice. The complete set of constraints consists of six primary ones and four secondary ones. The primary constraints generate gauge transformations on the phase space interpreted as Lorentz transformations \cite{oko-tegr-I} and can be divided into boost and rotation constraints. The secondary constraints are a scalar and vector constraints---the scalar constraint is a fairly complicated function of the canonical variables. All the constraints are of the first class \cite{oko-tegr-II}.

We will show that every point of the phase space can be gauge-transformed by Lo\-ren\-tzian boosts to a point, where each function $\xi^I$ is the zero function on $\Sigma$. In the constraints there are many terms containing the functions and in the gauge $\xi^I=0$ most of them will be annihilated. Moreover, in this gauge the coframe $(\theta^I)$ becomes orthonormal with respect to the Riemannian metric induced on $\Sigma$ by the spacetime (Lorentzian) metric.     

It will be then natural to \emi express all tensor fields appearing in the constraints in terms of their components in the orthonormal coframe $(\theta^I)$ and \emii decompose the two-forms $(r_I)$ and $(d\theta^I)$ therein into irreducible representations of $SO(3)$, that is, into the antisymmetric, traceless symmetric and trace parts. In this way we will obtain another quite transparent form of the constraints. The direct benefit of this step will be very simple general solutions to both boost and rotation constraints: $\zeta_I$ will be completely fixed as a linear function of the antisymmetric part of $(d\theta^I)$, while the antisymmetric part of $(r_I)$ will be forced to vanish. Thus the set of all the constraints will be narrowed to the scalar and vector constraints, being conditions imposed on the one-forms $(\theta^I)$ and the two-forms $(r_I)$, the latter ones lacking its antisymmetric part. 

Moreover, the obtained form of the scalar and vector constraints will raise the following question: is it possible to use the remaining gauge freedom to annihilate the antisymmetric part of $(d\theta^I)$ (keeping at the same time the gauge $\xi^I=0$ intact) and simplify thereby the constraints even further? We will make a preliminary analysis of this issue---results of the analysis will suggest that this simplification is possible.

Finally, to illustrate the usefulness of the simplified form of the constraints we will derive some exact solutions to them. We will obtain three families of solutions---one of them will be labeled by solutions to an eikonal equation.

Let us emphasize that the present paper does not claim to be exhaustive---further analysis of the constraints and the solutions obtained here as well as new solutions to the constraints will be presented in the forthcoming paper \cite{prep-ao-js}.

\section{Preliminaries}

Let $(\mathbb{M},\eta)$ be a four-dimensional Minkowski vector space---here $\eta$ is a scalar product of signature $(-,+,+,+)$ defined on the four-dimensional vector space $\mathbb{M}$. We additionally assume that an orientation of $\mathbb{M}$ is fixed. Let us also fix a basis $(e_A)$, $A\in\{0,1,2,3\}$, of $\mathbb{M}$, orthonormal with respect to $\eta$ and compatible with the orientation of $\mathbb{M}$. Denote by $(\eta_{AB})$ components of $\eta$ in the basis and by $(\eta^{AB})$ the components of the inverse (dual) scalar product in the same basis. We will use these components to raise and lower indices denoted by initial letters of the Latin alphabet $A,B,C,D\in\{0,1,2,3\}$.         

Denote by $\mathbb{E}$ the linear span of $\{e_1,e_2,e_3\}$. Let us also fix that orientation of $\mathbb{E}$, which is compatible with its basis $(e_I)$, $I\in\{1,2,3\}$. $\eta$ restricted to $\mathbb{E}$ is a positive definite scalar product and the basis $(e_I)$ is obviously orthonormal with respect to the product. Its components in $(e_I)$ will be denoted by $(\delta_{IJ})$, the components of the inverse product by $(\delta^{IJ})$. We will use these components to raise and lower indices denoted by letters from the middle of the Latin alphabet $I,J,K,L,...\in\{1,2,3\}$. 

Suppose that $\Mc$ is an oriented four-dimensional manifold. The configuration space for a Lagrangian formulation of TEGR (see e.g. \cite{thir}) can be chosen to be the set of all non-degenerate one-forms on $\Mc$ valued in $\mathbb{M}$. An $\mathbb{M}$-valued  one-form $\boldsymbol{\theta}$ on $\Mc$ is non-degenerate, if its decomposition $\boldsymbol{\theta}=\boldsymbol{\theta}^A\ot e_A$ in the basis $(e_A)$ provides us with a global coframe (cotetrad) $(\boldsymbol{\theta}^A)$ on $\Mc$. This coframe defines on the manifold a Lorentzian metric
\begin{equation}
g=\eta_{AB}\,\boldsymbol{\theta}^A\ot\boldsymbol{\theta}^B,
\label{g}
\end{equation}
which makes $\Mc$ a spacetime of general relativity. In this formulation the action is a functional depending on the one-forms $(\boldsymbol{\theta}^A)$ and their exterior derivatives.

Since the Lagrangian formulation of TEGR is based on differential forms, it is very convenient to describe the corresponding Hamiltonian version of the theory in the same terms---this can be done by means of a formalism developed in \cite{ham-diff} and \cite{mielke} (see also \cite{os}). Then the Legendre transformation gives the following canonical variables \cite{oko-tegr-I}: a quadruplet of one-forms $(\theta^A)$ and a quadruplet of two-forms $(p_A)$ as momenta conjugate to $(\theta^A)$, all defined on an oriented three-dimensional slice of $\Mc$, which will be denoted by $\Sigma$ \footnote{In \cite{oko-tegr-I} we assumed that $\Mc=\R\times\Sigma$ and that the orientation of $\Sigma$ is induced by the orientation of $\Mc$ in the following way: suppose $(x^i)$ is a (local) coordinate system on $\Sigma$ and $t\in\R$. Then $(t,x^i)$ is a coordinate system on $\Mc$. If $(t,x^i)$ is compatible with the orientation of $\Mc$, then $(x^i)$ is compatible with the orientation of $\Sigma$.}. The slice is spatial in the sense that
\begin{equation}
q=\eta_{AB}\,\theta^A\ot\theta^B, 
\label{q}
\end{equation}
is a Riemannian metric on $\Sigma$ --- this metric is in fact induced on $\Sigma$ by the spacetime metric \eqref{g}.             

In \cite{q-suit,ham-nv} we introduced new non-standard variables on the phase space of TEGR:
\begin{enumerate}
\item a sextuplet $(\xi^I,\theta^J)$, $I,J\in\{1,2,3\}$, such that $\xi^I$ is a real function on $\Sigma$ and the one-forms $(\theta^J)$ constitute a global coframe on the manifold;
\item a sextuplet $(\zeta_I,r_J)$, $I,J\in\{1,2,3\}$, where $\zeta_I$ is a three-from and $r_J$ a two-form on $\Sigma$.     
\end{enumerate}
The three-form $\zeta_I$ is a momentum conjugate to $\xi^I$, while the two-form $r_J$ is a momentum conjugate to $\theta^J$.

To describe a relation between the new variables and the original ones $(\theta^A,p_B)$  let us introduce a symbol 
\begin{equation}
\sgn(\theta^I):=
\begin{cases}
1 & \text{if $(\theta^I)$ is compatible with the orientation of $\Sigma$}\\
-1 & \text{otherwise}
\end{cases}
\label{sgn}
\end{equation}
and express the spatial metric \eqref{q} in terms of $(\xi^I,\theta^J)$:  
\begin{equation}
q=q_{IJ}\,\theta^I\ot\theta^J=\Big(\delta_{IJ}-\frac{\xi_I\xi_J}{1+\xi^K\xi_K}\Big)\,\theta^I\ot\theta^J.
\label{q-xi}
\end{equation}
In a local coordinate frame $q=q_{ij}\,dx^i\ot dx^j$. The metric inverse (dual) to $q$ maps a one-form $\alpha$ on $\Sigma$ to a vector field on $\Sigma$, which will be denoted by $\vec{\alpha}$. Thus in the local coordinate system
\[
\vec{\alpha}=q^{ij}\alpha_i\partial_{x^j},
\]    
where $(q^{ij})$ is the inverse of $(q_{ij})$. We have
\begin{equation}
\vth{}^I\lr\theta^J=\delta^{IJ}+\xi^I\xi^J. 
\label{tt-q}
\end{equation}

The original variables depend on the new ones in the following way\footnote{In \cite{q-suit,ham-nv} we introduced in fact a family of new variables labeled by functions defined on the set of all coframes $(\theta^I)$ and valued in the set $\{-1,1\}$. Here we use the variables given by the function $(\theta^I)\mapsto \sgn(\theta^I)$ defined by \eqref{sgn}.}: 
\begin{align}
\theta^0&=\sgn(\theta^J)\frac{\xi_I}{\sqrt{1+\xi_K\xi^K}}\,\theta^I, & \theta^I&=\theta^I,\label{th0}\\
p_0&=\sgn(\theta^J)\sqrt{1+\xi_K\xi^K}\vec{\theta}^I\lr\zeta_I, & p_I&=r_I-\xi_I\,\vec{\theta}^J\lr\zeta_J.\nonumber
\end{align}
In terms of the new variables the Poisson bracket on the phase space reads\footnote{Suppose that $\alpha$ is a $k$-form on $\Sigma$, and $F$ a functional, which depends on $\alpha$ and $d\alpha$. Then functional derivative $\delta F/\delta\alpha$ is a $(3-k)$-form such that for every variation $\delta\alpha$
\[
\delta F=\int_\Sigma\delta\alpha\we \frac{\delta F}{\delta\alpha}.
\]
\label{funct-der}}
\[
\{F,G\}=\int_\Sigma\Big(\frac{\delta F}{\delta {\xi}^I}\we\frac{\delta G}{\delta \zeta_I}+\frac{\delta F}{\delta {\theta}^I}\we\frac{\delta G}{\delta r_I}-\frac{\delta G}{\delta {\xi}^I}\we\frac{\delta F}{\delta \zeta_I}-\frac{\delta G}{\delta {\theta}^I}\we\frac{\delta F}{\delta r_I}\Big).
\]

The physical part of the phase space is distinguished by a number of constraints---below we will present smeared versions of them. In the formulas below we will use the Hodge dualization ${\hd}$ defined by the metric $q$ and the orientation of $\Sigma$, as well as a symbol
\begin{equation}
\xi^0:=\sgn(\theta^I)\sqrt{1+\xi^K\xi_K}.
\label{xi^0}
\end{equation}
In the smeared form there are two primary constraints, the boost and the rotation constraint respectively:
\begin{align}
B(a)&=\int_\Sigma a\we\Big(-\frac{\xi_Iq_{JK}}{(\xi^0)^2}\theta^I\we{\hd}(d\xi^J\we\theta^K)+q_{IJ}\theta^I\we{\hd}d\theta^J+\vec{\theta}^I\lr\zeta_I+\xi^Ir_I\Big), \label{boost-c}\\
R(b)&=\int_\Sigma b\we(\theta^I\we{\hd}r_I+q_{IJ}d\xi^I\we\theta^J).\label{rot-c}
\end{align}
and two secondary ones: the vector constraint
\begin{equation}
V(\vec{N})=\int_\Sigma d\xi^I\we\vec{N}\lr\zeta_I -{d}{\theta}^I\we(\vec{N}\lr r_I)-(\vec{N}\lr{\theta}^I)\we {d}r_I
\label{v(n)-new}
\end{equation}
and the scalar one
\begin{multline}
S(N)=\int_\Sigma N\Big(\frac{1}{2}(r_I\we\theta^J)\we{\hd}(r_J\we\theta^I)-\frac{1}{4}(r_I\we\theta^I)\we{\hd}(r_J\we\theta^J)-d(\vec{\theta}^J\lr\zeta_J)-\xi^I\we dr_I+\\+\frac{1}{4(\xi^0)^4}d(\xi_I\theta^I)\we\xi_J\theta^J\we{\hd}(d(\xi_K\theta^K)\we\xi_L\theta^L)+\frac{1}{2(\xi^0)^2}d(\xi_I\theta^I)\we\xi_J\theta^J\we{\hd}(d\theta_K\we\theta^K)-\\-\frac{1}{(\xi^0)^{2}}(q_{IJ}d\xi^J\we\theta^I+\xi_Id\theta^I)\we\theta^K\we{\hd}(d\theta_K\we\xi_L\theta^L)+\frac{1}{2}d\theta_I\we\theta^J\we{\hd}(d\theta_J\we\theta^I)-\\-\frac{1}{4}d\theta_I\we\theta^I\we{\hd}(d\theta_J\we\theta^J)\Big).
\label{scal-c}
\end{multline}
In the formulas describing the constraints $a$ and $b$ are one-forms on $\Sigma$, $N$ is a function and $\vec{N}$ a vector field on the manifold. All these fields play a role of smearing fields. All the constraints above turned out to be of the first class---this fact was proven in \cite{oko-tegr-II}, where the constraints were expressed in terms of the original variables $(\theta^A,p_B)$.

\section{Gauge $\xi^I=0$}

Since in the constraints \eqref{boost-c}--\eqref{scal-c} there are many terms containing the functions $(\xi^I)$, it is clear that on the subset of the phase space, where the functions are zero, the constraints are much simpler. Moreover, Equation \eqref{q-xi} implies that if $\xi^I=0$, then the coframe $(\theta^I)$ is orthonormal with respect to the spatial metric $q$.  On the other hand it was demonstrated in \cite{oko-tegr-I} that the boost and rotation constraints, \eqref{boost-c} and \eqref{rot-c}, generate gauge transformations on the phase space, which can be interpreted as local Lorentz transformations. Hence a question arises: is it possible to use these gauge transformations to zero the variables $(\xi^I)$? 

In this section we will show that the answer to this question is in affirmative---the variables $(\xi^I)$ can be always gauge-transformed to zero on the whole $\Sigma$ by means of gauge transformations generated by the boost constraint.

To be more precise, let us fix an arbitrary point $(\zeta^0_I,r^0_J,\xi^K_0,\theta^L_0)$ in the phase space and denote by $\Omega^1({\Sigma})$ the space of all one-forms on the manifold $\Sigma$. We will show that there exists a smooth\footnote{The curve \eqref{l-al} is smooth in the sense that for every point $y\in \Sigma$ the corresponding curve $]\lambda_1,\lambda_2[\,\ni\lambda\mapsto a(\lambda)_y\in T^*_y\Sigma$, where $a(\lambda)_y$ is the value of $a(\lambda)$ at $y$, is smooth.} curve 
\begin{equation}
\R\supset\,]\lambda_1,\lambda_2[\,\ni\lambda\to a(\lambda)\in\Omega^1({\Sigma})
\label{l-al}
\end{equation}
such that \emi for every $\lambda$ the one-form $a(\lambda)$ does not depend on the canonical variables and \emii  the gauge transformation generated by the boost constraint $B(a(\lambda))$ maps the point $(\zeta^0_I,r^0_J,\xi^K_0,\theta^L_0)$ to a point 
\begin{equation}
({\zeta}_I,{r}_J,0,{\theta}^L).
\label{ph-xi=0}
\end{equation}
Moreover, for smooth fields $(\xi^K_0,\theta^L_0)$ the corresponding map \eqref{l-al} will be valued in smooth one-forms.

\subsection{Gauge transformations given by the boost constraint \label{gau-boo}}

Given a map \eqref{l-al} of values independent of canonical variables, the gauge transformation generated by $B(a(\lambda))$ is given by the following differential equation system:
\begin{align}
\frac{d\zeta_I}{d\lambda}&=\{\zeta_I,B(a(\lambda))\}=-\frac{\delta B(a(\lambda))}{\delta \xi^I},\label{zeta-l}\\
\frac{dr_J}{d\lambda}&=\{r_J,B(a(\lambda))\}=-\frac{\delta B(a(\lambda))}{\delta \theta^J},\label{r-l}\\
\frac{d\xi^K}{d\lambda}&=\{\xi^K,B(a(\lambda))\}=\frac{\delta B(a(\lambda))}{\delta \zeta_K}=\vec{\theta}^K\lr a(\lambda),\label{xi-l}\\
\frac{d\theta^L}{d\lambda}&=\{\theta^L,B(a(\lambda))\}=\frac{\delta B(a(\lambda))}{\delta r_L}=\xi^L a(\lambda),\label{theta-l}
\end{align}  
---see Footnote \ref{funct-der} for a definition of the functional derivatives.

These gauge transformations are defined on the whole phase space, but we will be interested merely in the action of the transformations on the fixed point $(\zeta^0_I,r^0_J,\xi^K_0,\theta^L_0)$. Let us then denote by
\begin{equation}
\lambda\mapsto\big(\tilde{\zeta}_I(\lambda),\tilde{r}_J(\lambda),\tilde{\xi}^K(\lambda),\tilde{\theta}^L(\lambda)\big)
\label{l-phs}
\end{equation}
a curve in the phase space being a solution of the equation system \eqref{zeta-l}--\eqref{theta-l} and satisfying the initial condition\footnote{The initial condition makes sense if $\lambda_1<0<\lambda_2$ in \eqref{l-al}, which can be assumed without loss of generality.} 
\begin{equation}
\Big(\tilde{\zeta}_I(0),\tilde{r}_J(0),\tilde{\xi}^K(0),\tilde{\theta}^L(0)\Big)=(\zeta^0_I,r^0_J,\xi^K_0,\theta^L_0).
\label{ini}
\end{equation}
Thus each value of \eqref{l-phs} for $\lambda\neq 0$ is the result of the gauge transformation acting on $(\zeta^0_I,r^0_J,\xi^K_0,\theta^L_0)$. 

Let us finally emphasize that Equations \eqref{zeta-l}--\eqref{theta-l} are defined {\em pointwisely} with respect to $\Sigma$, i.e., at any point $y\in\Sigma$ the l.h.s. of the equations are derivatives of field values at $y$ and the r.h.s. are built from field values at the same point. However, in the analysis below in order to keep the notation simple, we will not use any special symbols to distinguish values of fields at $y$ from the fields. This concerns Sections \ref{gtr-xi}--\ref{gtr-mom} except those parts of them, where it will be clear from the context that symbols used denote fields rather than their values at a point.    

\subsection{Choice of $a(\la)$ }

In order to transform the point $(\zeta^0_I,r^0_J,\xi^K_0,\theta^L_0)$ to one, where $\xi^I=0$, we have to select a specific gauge transformation from the set given by Equations \eqref{zeta-l}--\eqref{theta-l}. Since this set is labeled by the curves \eqref{l-al}, we will make the selection by choosing a curve $\la\mapsto a(\lambda)$ adapted appropriately to the point $(\zeta^0_I,r^0_J,\xi^K_0,\theta^L_0)$. To describe our choice let us note that for every $\la$ the variables $\big(\tilde{\theta}^I(\la)\big)$ in \eqref{l-phs} form a global coframe on $\Sigma$. Our choice of an adapted curve is
\begin{equation}
\la\mapsto a(\lambda)=a_J \tilde{\theta}^J(\la),
\label{al-th}
\end{equation}
where $(a_J)$ are real functions on $\Sigma$ given by
\begin{equation}
a_J=
\begin{cases}
0 & \text{if $\xib_0=0$},\\
-\dfrac{\xi_{0J}}{\xib_0}\arctg\xib_0 & \text{otherwise}.
\end{cases}
\label{a_J}
\end{equation}
In the formula above
\[
\xib_0:=\sqrt{\delta_{IJ}\xi^I_0\xi^J_{0}}\geq 0.
\]

Let us now comment on the choice of $a(\la)$. Deriving Equations \eqref{zeta-l}--\eqref{theta-l} we assumed that the one-form $a(\la)$ does not depend on the canonical variables, but our choice \eqref{al-th} depend on $(\xi^I_0)$ and $\tilde{\theta}^I(\la)$. Is there any discrepancy here? To answer this question let us note that any gauge transformation described by the equations is defined on the whole phase space and the assumed independence of $a(\la)$ of the canonical variables means that $a(\la)$ is the same for every point of the phase space. Our one-form \eqref{al-th} is adapted to the fixed point $(\zeta^0_I,r^0_J,\xi^K_0,\theta^L_0)$ and in this fixed form, that is, being the same for every point of the phase space, defines a gauge transformation on the whole space. Thus there is no discrepancy here.

\subsection{Gauge transformation of $\xi^I_0$ \label{gtr-xi}}

Now let us fix a point $y\in\Sigma$ and proceed as described in the last remark in Section \ref{gau-boo}. Inserting $a(\la)$ given by \eqref{al-th} to Equation \eqref{xi-l} and applying \eqref{tt-q} we obtain
\begin{equation}
\frac{d\tilde{\xi}^K}{d\lambda}=a_J\vec{\tilde{\theta}}^K\lr\tilde{\theta}^J=a_J(\delta^{KJ}+\tilde{\xi}^K\tilde{\xi}^J)=
\begin{cases}
0 & \text{if $\xib_0=0$},\\
-\dfrac{\xi_{0J}}{\xib_0}\arctg\xib_0\, (\delta^{KJ}+\tilde{\xi}^K\tilde{\xi}^J)& \text{otherwise}.
\end{cases}
\label{xi-lxx}
\end{equation}
It is a simple exercise to convince oneself that 
\begin{equation}
\tilde{\xi}^K(\lambda)=
\begin{cases}
0 & \text{if $\xib_0=0$},\\
\dfrac{\xi^K_{0}}{\xib_0}\tg\big((1-\lambda)\arctg\xib_0\big) & \text{otherwise},
\end{cases}
\label{tl-xi-l}
\end{equation}
satisfies \eqref{xi-lxx} for every $\lambda$, for which the tangent is well defined. For our purpose it is enough to restrict the range of $\lambda$ in such a way that the argument of the tangent belongs to $]-\pi/2,\pi/2[$. Then 
\begin{equation}
\lambda\in
\begin{cases}
\R & \text{if $\xib_0=0$},\\
\big]1-\dfrac{\pi}{2\arctg\xib_0},1+\dfrac{\pi}{2\arctg\xib_0}\big[ & \text{otherwise}.
\end{cases}
\label{l-range}
\end{equation}
It is clear, that for every $y\in\Sigma$ the range of $\lambda$ includes the interval $[0,1]$ and consequently
\begin{align*}
&\tilde{\xi}^K(0)=\xi^K_0,&&\tilde{\xi}^K(1)=0
\end{align*}
everywhere on $\Sigma$.

\subsection{Gauge transformations of $\theta^L_0$}

To find the associated gauge transformation of $\theta^L_0$ we rewrite \eqref{theta-l} as
\[
\frac{d\tilde{\theta}^L}{d\lambda}=\tilde{\xi}^L(\lambda) a_J\tilde{\theta}^J,
\]  
where $\tilde{\xi}^L(\lambda)$ is given by \eqref{tl-xi-l} and $a_J$ by \eqref{a_J}, hence
\begin{equation}
\frac{d\tilde{\theta}^L}{d\lambda}=
\begin{cases}
0 & \text{if $\xib_0=0$},\\
-\dfrac{\tg\big((1-\lambda)\arctg\xib_0\big)}{\xib_0} \dfrac{\arctg\xib_0}{\xib_0}\,\xi^L_{0}\xi_{0J}\,\tilde{\theta}^J & \text{otherwise}.
\end{cases}
\label{theta-l-1}
\end{equation}
Let $\Xi$ be a function on $\Sigma$ valued in the space of all linear operators on $\mathbb{E}$ such that
\[
\Xi^L{}_J=\xi^L_{0}\xi_{0J}.
\]    
Then the solution to Equation \eqref{theta-l-1} corresponding to non-zero $\xib_0$ is given by
\begin{multline}
\tilde{\theta}^L(\lambda)={\bigg[\exp\Big[\Big(\int_0^\lambda-\tg\big((1-\lambda')\arctg\xib_0\big)\arctg \xib_0\,d\lambda'\Big)\,\xib_0^{-2}\Xi \Big]\bigg]^L}_J\,\theta^J_0=\\=
{\bigg[\exp\Big[\ln\Big(\frac{\cos(\arctg\xib_0)}{\cos((1-\lambda)\arctg\xib_0)}\Big)\,\xib_0^{-2}\Xi \Big]\bigg]^L}_J\,\theta^J_0
\label{t-th-sol}
\end{multline}
(note that this formula confirms that $\big(\tilde{\theta}^L(\lambda)\big)$ is a coframe on $\Sigma$). For non-zero $\xib_0$ 
\[
(\xib_0^{-2}\Xi)^2=\xib_0^{-2}\Xi,
\]
hence
\[
(\xib_0^{-2}\Xi)^n=\xib_0^{-2}\Xi,
\]
for every natural $n>0$. Denoting by $f(\lambda)$ the logarithm at the second line of \eqref{t-th-sol} and by $\mathbf{1}$ the identity map on $\mathbb{E}$, we calculate
\begin{multline*}
\exp\big(f(\lambda)\xib_0^{-2}\Xi\big)=\mathbf{1}+\sum_{n=1}^\infty\frac{\big(f(\lambda)\xib_0^{-2}\Xi\big)^n}{n!}=\mathbf{1}+\sum_{n=1}^\infty\frac{(f(\lambda))^n}{n!}\xib_0^{-2}\Xi=\\=\mathbf{1}+\Big(-1+1+\sum_{n=1}^\infty\frac{(f(\lambda))^n}{n!}\Big)\xib_0^{-2}\Xi=\mathbf{1}+\big(\exp(f(\lambda))-1\big)\xib_0^{-2}\Xi=\\=\mathbf{1}+\Big(\frac{\cos(\arctg\xib_0)}{\cos((1-\lambda)\arctg\xib_0)}-1\Big)\xib_0^{-2}\Xi.
\end{multline*}  
Thus the solution of \eqref{theta-l-1} reads
\begin{equation}
\tilde{\theta}^L(\lambda)=
\begin{cases}
\theta^L_0 & \text{if $\xib_0=0$},\\
\theta^L_0+\dfrac{\cos(\arctg\xib_0)-\cos((1-\lambda)\arctg\xib_0)}{\xib_0^2\cos((1-\lambda)\arctg\xib_0)}\xi^L_0\xi_{0J} \theta^J_0 & \text{otherwise}.
\end{cases}
\label{t-th-fin}
\end{equation}
Clearly, it satisfies the initial condition 
\[
\tilde{\theta}^L(0)=\theta^L_0.
\]
and the range of $\lambda$, for which the solution \eqref{t-th-fin} is valid, is given by \eqref{l-range}.

\subsection{Gauge transformations of the momenta \label{gtr-mom}}

\subsubsection{Functional derivative of the Hodge operator}

Since the boost constraint \eqref{boost-c} contains terms built by means of the Hodge operator ${\hd}$ the r.h.s. of both Equations \eqref{zeta-l} and \eqref{r-l} contain functional derivatives of the operator with respect to $\xi^K$ and $\theta^L$. To proceed further with the analysis of the gauge transformations we have to convince ourselves that the derivatives are well defined. 

Given $k$-forms $\alpha$ and $\beta$ on $\Sigma$ ($0\leq k\leq 3$), let us introduce a symbol
\begin{equation}
\alpha \we\Star{A}\beta\equiv\vec{\theta}^B\lr\Big(\eta_{AB}\,\alpha\we{{\hd}}\beta-(\vec{\theta}_A\lr\alpha)\we{{\hd}}(\vec{\theta}_B\lr\beta)-(\vec{\theta}_B\lr\alpha)\we{{\hd}}(\vec{\theta}_A\lr\beta)\Big).
\label{a*'b}
\end{equation}
It was shown in \cite{os} that if $\alpha,\beta$ are independent of the variables $(p_A,\theta^B)$ and 
\[
F=\int_\Sigma\alpha\we{\hd}\beta
\]
is a functional depending on $\theta^A$ via the Hodge operator, then 
\begin{equation}
\frac{\delta F}{\delta\theta^A}=\alpha\we\Star{A}\beta.
\label{dFdth}
\end{equation}
By definition of the functional derivative with respect to a differential form (see Footnote \ref{funct-der})
\begin{equation}
\delta F=\int_\Sigma \delta\theta^A\we \frac{\delta F}{\delta\theta^A}.
\label{dF-df}
\end{equation}
Using \eqref{th0} it is easy to check that
\[
\delta\theta^0=\frac{q_{IJ}}{\xi^0}(\delta\xi^I)\theta^J+\frac{\xi_I}{\xi^0}\delta\theta^I,
\]
where $(q_{IJ})$ are given by \eqref{q-xi} and $\xi^0$ by \eqref{xi^0}. Setting this result and \eqref{dFdth} to Equation \eqref{dF-df} we obtain
\[
\delta F=\int_\Sigma \Big[\delta\xi^I\we\Big( \frac{q_{IJ}}{\xi^0}\theta^J\we (\alpha\we\Star{0}\beta)\Big)+\delta\theta^I\Big(\frac{\xi_I}{\xi^0}(\alpha\we\Star{0}\beta)+(\alpha\we\Star{I}\beta)\Big)\Big]. 
\]
Thus
\begin{align*}
&\frac{\delta F}{\delta\xi^I}=\frac{q_{IJ}}{\xi^0}\theta^J\we (\alpha\we\Star{0}\beta)\equiv \alpha \we \Starxi{I}\beta,\\
&\frac{\delta F}{\delta\theta^I}=\frac{\xi_I}{\xi^0}(\alpha\we\Star{0}\beta)+(\alpha\we\Star{I}\beta)\equiv \alpha\we \Starth{I}\beta.
\end{align*}

Taking into account the formula \eqref{a*'b}, we conclude that for every $(\xi^I,\theta^J)$ both functional derivatives above are well defined. Of course, one can find explicit expressions for the derivatives in terms of $(\xi^I,\theta^I)$, but we will not need them, since the conclusion will turn out to be sufficient for our goal.

\subsubsection{Gauge transformations of $\zeta^0_I$ and $r^0_J$}

Suppose that $\alpha$ is a one-form and $\beta$ a $k$-form on $\Sigma$. Then \cite{os}
\begin{equation}
{\hd}({\hd}\beta\we\alpha)=\vec{\alpha}\lr\beta.
\label{a-b}
\end{equation}

Denoting
\[
B'(a)\equiv\int_\Sigma a\we\Big(-\frac{\xi_Iq_{JK}}{(\xi^0)^2}\theta^I\we{\hd}(d\xi^J\we\theta^K)+q_{IJ}\theta^I\we{\hd}d\theta^J\Big)
\]
and using Equation \eqref{a-b} the boost constraint \eqref{boost-c} can be rewritten as follows:
\[
B(a)=B'(a)+\int_\Sigma  a\we\big( {\hd}({\hd}\zeta_I\we\theta^I)+\xi^Ir_I\big)=B'(a)+\int_\Sigma \big( \theta^I\we{\hd}a\we {\hd}\zeta_I+\xi^Ia\we r_I\big) 
\]
---in the last step we used the fact that ${\hd}\zeta_I$ is a zero-form (a function). This formula allows us to transform Equations \eqref{zeta-l} and \eqref{r-l} to the following form:
\begin{equation}
\begin{aligned}
\frac{d\tilde{\zeta}_I}{d\lambda}&=-\frac{\delta B'(a(\lambda))}{\delta \xi^I}-(\tilde{\theta}^K\we\Starxi{I}a(\lambda)){\hd}\tilde{\zeta}_K-\big( (\tilde{\theta}^K\we{\hd}a(\lambda))\we \Starxi{I}\tilde{\zeta}_K\big)-a(\lambda)\we \tilde{r}_I,\\
\frac{d\tilde{r}_J}{d\lambda}&=-\frac{\delta B'(a(\lambda))}{\delta \theta^J}-{\hd}a(\lambda)({\hd}\tilde{\zeta}_J)-(\tilde{\theta}^L\we\Starth{J}a(\lambda))({\hd}\tilde{\zeta}_L)-\big( (\tilde{\theta}^L\we{\hd}a(\lambda))\we \Starth{J}\tilde{\zeta}_L\big).
\end{aligned}
\label{mom-l}
\end{equation}
Let us emphasize that the r.h.s. of these two equations are functions of $a(\lambda)$ given by \eqref{al-th}, $\tilde{\theta}^I(\lambda)$ given by \eqref{t-th-fin} and $\tilde{\xi}^I(\lambda)$ described by \eqref{tl-xi-l}. Because of their peculiar form, it is not clear that these fields are differentiable at those points of $\Sigma$, where $\xib_0$ vanishes. Non-differentiability of the fields would be troublesome, because the functional $B'(a)$ contains the derivatives $d\xi^I$ and $d\theta^J$. Fortunately, all these fields are smooth provided $(\xi^I_0,\theta^J_0)$ are smooth---a justification of this statement will be given in Section \ref{smooth}.          

Thus we obtained a system \eqref{mom-l} of differential equations. Of course, we do not dare to try to solve it, but we do not need an explicit solution actually. In fact, it is sufficient to show that there exists a solution $\lambda\mapsto(\tilde{\zeta}_I(\lambda),\tilde{r}_J(\lambda))$ satisfying the appropriate initial condition (see \eqref{ini}) and defined on the interval \eqref{l-range}. To this end we will apply the following lemma \cite{maurin}\footnote{The original lemma in \cite{maurin} is formulated in terms of a general Banach space rather than a finite dimensional vector space.}:

\begin{lm}
Let $V$ be a finite-dimensional real vector space and $L(V)$ a space of linear maps from $V$ into itself. If mappings $[\Lambda_1,\Lambda_2]\ni\lambda\mapsto A(\lambda)\in L(V)$ and $[\Lambda_1,\Lambda_2]\ni\lambda \mapsto b(\lambda)\in V$ are continuous on $[\Lambda_1,\Lambda_2]$, then for every $(\lambda_0,v_0)\in [\Lambda_1,\Lambda_2]\times V$ there exists exactly one solution $[\Lambda_1,\Lambda_2]\ni\lambda\mapsto v(\lambda)\in V$ of the differential equation
\begin{equation}
\frac{dv}{d\lambda}=A(\lambda)v+b(\lambda),
\label{ilode}
\end{equation}         
which satisfies the initial condition $v(\lambda_0)=v_0$. 
\label{ilode-lm}
\end{lm}
\noindent In other words, the lemma guarantees the existence and uniqueness of a solution of the {\em non-homogeneous linear} differential equation \eqref{ilode}, satisfying a fixed initial condition. Moreover, the solution is defined on the {\em entire} interval $[\Lambda_1,\Lambda_2]$. Let us show now that the lemma is applicable to the system \eqref{mom-l}.

Given $y\in\Sigma$, all collections $(\tilde{\zeta}_I,\tilde{r}_J)$ form a $12$-dimensional real vector space $V$. Since the functional $B'(a)$ is independent of the momenta the functional derivatives at the r.h.s. of \eqref{mom-l} constitute the non-homogeneous term corresponding to $b(\lambda)$ in \eqref{ilode}. It is clear that the remaining parts of the r.h.s. of \eqref{mom-l} are linear in $(\tilde{\zeta}_I,\tilde{r}_J)$ and consequently they correspond to the term $A(\lambda)v$ in \eqref{ilode}. In other words, for every $y\in\Sigma$ the system \eqref{mom-l} is a non-homogeneous linear differential equation defined on a finite-dimensional vector space. Taking into account that the functional derivatives of the Hodge operator ${\hd}$ with respect to $\xi^K$ and $\theta^L$ are well defined (see the previous subsection), it is not difficult to realize that both $(i)$ the non-homogeneous term in \eqref{mom-l} and $(ii)$ the linear operator acting on $(\tilde{\zeta}_I,\tilde{r}_J)$ in \eqref{mom-l} are continuous functions of the parameter $\lambda$ on the entire (open) interval \eqref{l-range}. 

We conclude that Lemma \ref{ilode-lm} is applicable to the system \eqref{mom-l} with the range of $\lambda$ restricted to every closed interval $[\Lambda_1,\Lambda_2]$ contained in \eqref{l-range}. Therefore, given $y\in\Sigma$,  there exists a (unique) solution $\lambda\mapsto(\tilde{\zeta}_I(\lambda),\tilde{r}_J(\lambda))$ of \eqref{mom-l} satisfying the initial condition 
\[
\big(\tilde{\zeta}_I(0),\tilde{r}_J(0)\big)=(\zeta^0_{I},r^0_{J})
\]                
and defined on the entire open interval \eqref{l-range}. 

\subsection{Summary of the analysis of gauge transformations}

Thus we showed that for every $y\in\Sigma$ and for every values of $(\zeta^0_{I},r^0_{J},\xi^K_0,\theta^L_0)$ at $y$, there exists a solution of the system \eqref{zeta-l}--\eqref{theta-l} with the specially chosen one-form $a(\lambda)$ such that \emi it is defined on the interval \eqref{l-range}, \emii it satisfies  the initial condition \eqref{ini} and \emiii $\tilde{\xi}^K(1)=0$. 

Given point $(\zeta^0_{I},r^0_{J},\xi^K_0,\theta^L_0)$ in the phase space, the ``pointwise'' solutions found in the previous subsections form a curve in the phase space, provided for every $y\in\Sigma$ there exists a non-empty interval $[\lambda_1,\lambda_2]$ contained in \eqref{l-range}. Clearly, such an interval exists and for our purposes we can choose it to be $[0,1]$.

Thus for every $(\zeta^0_{I},r^0_{J},\xi^K_0,\theta^L_0)$ there exists a curve \eqref{l-phs} in the phase space such that
\begin{enumerate}
\item it is defined on the interval $[0,1]$;
\item it is a solution of the system \eqref{zeta-l}--\eqref{theta-l} specified by the curve \eqref{l-al} defined on the interval $[0,1]$ by Equation \eqref{al-th};
\item it satisfies  the initial condition \eqref{ini};
\item $\tilde{\xi}^K(1)=0$. 
\end{enumerate}

All these mean that every point of the phase space can be gauge-transformed to one, for which $\xi^K=0$ everywhere on $\Sigma$.

\section{Smoothness of the gauge-transformed fields $(\tilde{\xi}^K)$ and $(\tilde{\theta}^L)$  \label{smooth} }

It is clear that the fields $\tilde{\xi}^K(\lambda)$ and $\tilde{\theta}^L(\lambda)$ given by, respectively, \eqref{tl-xi-l} and \eqref{t-th-fin}, are smooth at every $y\in\Sigma$ such that $\xib_0(y)\neq 0$, provided the initial fields $(\xi_0^I)$ and $(\theta^J_0)$ are smooth. But if $\xib_0(y)=0$ one may expect non-smoothness or even non-differentiability of these fields at such $y$, because 
\begin{enumerate}
\item the fields are described by two different formulae depending on whether $\xib_0$ is zero or not;
\item there is $\xib_0$ in denominators of some expressions there;
\item the function 
\[
\Sigma\ni y\mapsto \xib_0(y)=\sqrt{\delta_{IJ}\xi^I(y)\xi^J(y)}\in\R
\]          
may be not differentiable at $y$, for which $\xib_0(y)=0$;
\end{enumerate} 
Below we will show that, contrary to this expectation, both fields are smooth at every $y$ such that $\xib_0(y)=0$.

Let us fix $\lambda\in\, ]0,1]$ and rewrite the solution \eqref{tl-xi-l} as follows:
\begin{equation}
\tilde{\xi}^K(\lambda)=\kappa_\lambda(\xib_0)\xi^K_0,
\label{tlxi-kp}
\end{equation}
where
\begin{equation}
\kappa_\lambda(x):=
\begin{cases}
1-\lambda & \text{if $x=0$},\\
\dfrac{\tg\big((1-\lambda)\arctg x\big)}{x} & \text{otherwise}.
\end{cases}
\label{kappa_l}
\end{equation}
On the other hand, the solution \eqref{t-th-fin} can be expressed in the following form:
\begin{equation}
\tilde{\theta}^L(\lambda)=\theta^L_0+\frac{\chi_\lambda(\xib_0)}{\cos\big((1-\lambda)\arctg(\xib_0)\big)}\,\xi^L_0\xi_{0J} \,\theta^J_0,
\label{tlth-chi}
\end{equation}
where
\[
\chi_\lambda(x):=
\begin{cases}
\dfrac{-\lambda^2+2\lambda}{2} & \text{if $x=0$},\smallskip\\
\dfrac{\cos((1-\lambda)\arctg x)-\cos(\arctg x)}{x^2}&\text{otherwise}.
\end{cases}
\]
Note that for $\lambda\in\,]0,1]$ and $\xib_0\geq 0$ the denominator in \eqref{tlth-chi} is never zero.

\begin{lm}
Let a smooth function $f:\R\to\R$ be analytic at $x=0$, i.e., there exists a positive number $\rho$ such that for every $|x|<\rho$
\begin{equation}
f(x)=\sum_{k=0}^\infty a_k x^k.
\label{pw-sr}
\end{equation}
If the function 
\[
h(x)=
\begin{cases}
a & \text{for $x=0$},\\
\dfrac{f(x)}{x^n} & \text{otherwise},
\end{cases} \quad \quad\text{$a\in \R$,\,\, $0<n\in\N$},   
\]
is continuous at $x=0$, then it is analytic at $x=0$ and smooth on $\R$.    
\label{f-analyt}
\end{lm}

\begin{proof}
Clearly, the function $h$ is smooth on $\R\setminus\{0\}$. If $h$ is continuous at $x=0$, then
\[
a=h(0)=\lim_{x\to 0}h(x)=\lim_{x\to 0}\frac{f(x)}{x^n}=\lim_{x\to 0} \sum_{k=0}^\infty a_k x^{k-n},
\]      
which means that $a_0=a_1=\ldots =a_{n-1}=0$ and $a_n=a$. Therefore for $|x|<\rho$
\[
h(x)=\sum_{k=0}^\infty a_{k+n}x^k
\]   
and, consequently, $h$ is analytic at zero and smooth on the whole $\R$. 
\end{proof}

\begin{lm}
Suppose $f$ satisfies the assumptions of Lemma \eqref{f-analyt}. If $f$ is an even function, then the function 
\[
[0,\infty[\ni x\mapsto \bar{f}(x):=f(\sqrt{x})\in \R
\]
is a restriction to $[0,\infty[$ of a smooth function $\check{f}:\,]-\rho^2,\infty[\,\to\R$.
\label{f-even}
\end{lm}

\begin{proof}
It is obvious that the function $\bar{f}$ is smooth on $]0,\infty[$. If $f$ is even, then in the power series \eqref{pw-sr} there is no term with an odd power of $x$. Thus for every $x$ such that $\sqrt{x}<\rho$ 
\[
\bar{f}(x)=\sum_{k=0}^\infty a_{2k}x^k.
\]  
But the power series above is convergent for every $|x|<\rho^2$. Therefore $\bar{f}$ can be analyticly extended to a smooth function $\check{f}$ defined on $]-\rho^2,\infty[\,$.
\end{proof}  

Consider now the function $\kappa_\lambda$ given by \eqref{kappa_l}. Since $\tg$ is the quotient of $\sin$ by $\cos$, it is an analytic function on $]-\pi/2,\pi/2[$ ($\cos$ in non-zero on the interval). Therefore $\arctg$ is analytic on $\R$ as the inverse of $\tg$. Finally, $x\mapsto \tg\big((1-\lambda)\arctg x\big)$ is analytic on $\R$ being a composition of two analytic functions. We  conclude then that the latter function satisfies all conditions Lemma \ref{f-analyt} imposes on the function $f$. Since $\kappa_\lambda$ is continuous at $x=0$, it is analytic at $x=0$ and smooth on $\R$ by virtue of the lemma.           

Note now that $\kappa_\lambda$ is an even function. This fact and the previous conclusion allow us to apply Lemma \ref{f-even} to the function. The lemma guarantees that there exists a smooth function $\check{\kappa}_\lambda$ defined on an open domain containing all non-negative real numbers such that
\[
\kappa_\lambda(\xib_0)=\check{\kappa}_\lambda(\xi_{0K}\xi^K_0).
\]    
Thus $\tilde{\xi}^K(\lambda)$ given by \eqref{tlxi-kp} is a smooth field on $\Sigma$, provided $\xi^K_0$ is smooth.

A similar reasoning applied to the functions $\xi^K_0\mapsto\chi_\lambda(\xib_0)$ and $\xi^K_0\mapsto \cos\big((1-\lambda)\arctg(\xib_0)\big)$, leads to the conclusion that $\tilde{\theta}^L(\lambda)$ given by \eqref{tlth-chi} is a smooth field on $\Sigma$, provided  $\xi^K_0$ and $\theta^L_0$ are smooth.

Let us also note that the formula \eqref{a_J} defining the components $a_J$, can be rewritten as   
\begin{equation}
a_J=\sigma(\xib_0)\,\xi_{0J},
\label{a_J-sig}
\end{equation}
where
\[
\sigma(x):=
\begin{cases}
-1& \text{if $x=0$},\\
-\dfrac{\arctg x}{x}&\text{otherwise}
\end{cases}.
\]
Now, using reasoning similar to the previous one and the smoothness of $\big(\tilde{\theta}^L(\lambda)\big)$ just proven, one can easily show that the one-form $a(\lambda)$ given by \eqref{al-th} is smooth, provided $(\xi_0^I)$ and $(\theta^J_0)$ are smooth.

\section{The constraints of TEGR in the gauge $\xi^I=0$ }

If $\xi^I=0$ everywhere on $\Sigma$, then according to \eqref{q-xi} $q_{IJ}=\delta_{IJ}$. Now the boost constraint \eqref{boost-c} can be rewritten in a form of the following algebraic condition imposed on the momentum $\zeta_I$:
\begin{equation}
\vth{}^I\lr\zeta_I+\theta_J\we {\hd}d\theta^J=0.
\label{boost-0}
\end{equation}
The rotation constraint \eqref{rot-c} simplifies to 
\[
\theta^I\we{\hd} r_I=0.
\]
Note that by virtue of \eqref{a-b}
\[
{\hd}(\theta^I\we{\hd} r_I)=-{\hd}({\hd}r_I\we\theta^I)=-\vth{}^I\lr r_I,
\]
which means that the rotation constraint can be put in the following form similar to \eqref{boost-0}:
\begin{equation}
\vth{}^I\lr r_I=0.
\label{rot-0}
\end{equation}

Since $(\theta^I)$ is a global coframe, the triplet $(\vth{}^I)$ is a global frame and therefore the vector field $\vec{N}$ in the vector constraint \eqref{v(n)-new} can be expressed as $\vec{N}=N_I\vth{}^I$. Then the constraint reads:
\[
d\theta^I\we N_J\vth{}^J\lr r_I+N_J\vth{}^J\lr\theta^I dr_I=N_J(d\theta^I\we \vth{}^J\lr r_I+\delta^{JI} dr_I)=0
\]  
(here we used \eqref{tt-q}). This condition holds for every $\vec{N}$. We can thus drop $N_J$ in the formula above obtaining a linear differential equation imposed on the momentum $r_I$: 
\begin{equation}
dr^J+d\theta^I\we \vth{}^J\lr r_I=0.
\label{vec-0}
\end{equation}

Regarding the scalar constraint \eqref{scal-c}, let us use the boost constraint \eqref{boost-0} to replace the term $\vth{}^I\lr \zeta_I$ by one being a function of $\theta^I$ only. Thus we get
\begin{multline*}
\frac{1}{2}(r_I\we\theta^J)\we{\hd}(r_J\we\theta^I)-\frac{1}{4}(r_I\we\theta^I)\we{\hd}(r_J\we\theta^J)+d(\theta_J\we {\hd}d\theta^J)+\frac{1}{2}d\theta_I\we\theta^J\we{\hd}(d\theta_J\we\theta^I)-\\-\frac{1}{4}d\theta_I\we\theta^I\we{\hd}(d\theta_J\we\theta^J)=0
\end{multline*}
---this expression can be slightly simplified by acting on it by the Hodge operator $\hd$. 

To summarize: by virtue of the gauge condition $\xi^I=0$ the constraints can be expressed in the following form: 
\begin{align}
&\begin{aligned}
&&\vth{}^I\lr\zeta_I+\theta_J\we {\hd}d\theta^J=0, &&\vth{}^I\lr r_I=0, && dr^J+d\theta^I\we \vth{}^J\lr r_I=0,
\end{aligned}\label{brv-c}\\
&[{\hd}(r_I\we\theta^J)][{\hd}(r_J\we\theta^I)]-\frac{1}{2}[{\hd}(r_I\we\theta^I)]^2+\nonumber\\+&2{\hd}d(\theta_J\we {\hd}d\theta^J)+[{\hd}(d\theta_I\we\theta^J)][{\hd}(d\theta_J\we\theta^I)]-\frac{1}{2}[{\hd}(d\theta_I\we\theta^I)]^2=0.
\label{*scal}
\end{align}
Comparing this form with the original one given by Equations \eqref{boost-c}--\eqref{scal-c} we see that, indeed, the gauge $\xi^I=0$ makes the constraints much simpler. Let us emphasize once again that now by virtue of the gauge and Equation \eqref{q-xi}, the coframe $(\theta^I)$ is orthonormal with respect to the spatial metric $q$.  

\section{The constraints in terms of irreducible parts of the variables}

The constraints \eqref{brv-c} and \eqref{*scal}, though simpler than the original ones, are still quite complex. Therefore looking for a way to simplify them further is worth the effort. The natural idea is to take advantage of the fact that in the gauge $\xi^I=0$ the coframe $(\theta^I)$ is orthonormal with respect to the metric $q$ and to express tensor fields appearing in the constraints in terms of their components in this coframe. Moreover, the two-forms $(r^J,d\theta^K)$ can be easily decomposed into irreducible representations of the $SO(3)$ group preserving the metric $q$. Below we will combine this two ideas to cast the constraints in an other form---this will give us general solutions to the boost and rotation constraints and will highlight an algebraic structure of the scalar constraint.         

\subsection{Preliminaries}

As it is said above, we are going to express some tensor fields in terms of their components in the coframe $(\theta^I)$. Consequently, below the components will be labeled by capital Latin indices $I,J,K$ etc.

Let $\eps$ be the volume form on $\Sigma$ given by the spatial metric $q$ and the orientation of the manifold. Since in the gauge $\xi^I=0$ the coframe $(\theta^I)$ is orthonormal with respect to $q$,  the components $(\eps_{IJK})$ of the volume form in the coframe satisfy\footnote{Note that our convention is to use $(\delta^{IJ})$ to raise lower indices $I,J,K,\ldots$. But if $\xi^I=0$, then $(\delta^{IJ})$ are the components of the metric inverse (dual) to $q$ in the coframe $(\theta^I)$. Therefore the formulas \eqref{ee} holds true.}
\begin{align}
\eps^{IJK}\eps_{LMN}&=3!\,\delta^{[I}{}_L\delta^J{}_M\delta^{K]}{}_N, & \eps^{IJK}\eps_{LMK}&=2!\,\delta^{[I}{}_L\delta^{J]}{}_M, & \eps^{IJK}\eps_{LJK}&=2!\,\delta^{I}{}_L
\label{ee}
\end{align}
and
\begin{equation}
\eps_{123}=\sgn(\theta^K).
\label{e-sgn}
\end{equation}
Moreover,
\begin{align}
{\hd}(\theta^I\we\theta^J\we\theta^K)&=\eps^{IJK}, & {\hd}(\theta^I\we\theta^J)&=\eps^{IJ}{}_{K}\theta^K.
\label{star-thetas}
\end{align}

\paragraph{The momenta $(\zeta_I)$} The momentum $\zeta_I$ conjugate to $\xi^I$ is a three-form. Therefore 
\begin{equation}
\zeta_I=\zeta_{IJKL}\,\theta^J\ot\theta^K\ot\theta^L=\frac{1}{3!}\zeta_{IJKL}\,\theta^J\we\theta^K\we\theta^L,
\label{zeta-th}
\end{equation}
where $\zeta_{IJKL}=\zeta_{I[JKL]}$. Using the first equation in \eqref{ee} we obtain
\begin{equation}
\zeta_{IJKL}=\zeta_{IMNP}\,\delta^{[M}{}_{J}\delta^{N}{}_{K}\delta^{P]}{}_{L}=\frac{1}{3!}\zeta_{IMNP}\eps^{MNP}\eps_{JKL}=\frac{1}{3!}\bar{\zeta}_I\eps_{JKL},
\label{zeta-ijkl}
\end{equation}
where
\begin{equation*}
\bar{\zeta}_I:=\zeta_{IMNP}\eps^{MNP}.
\end{equation*}

\paragraph{The momenta $(r_I)$} The momentum $r_I$ conjugate to $\theta^I$ is a two-form. Therefore 
\begin{equation}
r^I=r^{I}{}_{JK}\,\theta^J\ot\theta^K=\frac{1}{2}r^I{}_{JK}\,\theta^J\we\theta^K,
\label{rI-comp}
\end{equation}
where $r^I{}_{JK}=r^I{}_{[JK]}$. Using the second equation in \eqref{ee} we obtain
\begin{equation}
r^I{}_{JK}=r^I{}_{MN}\,\delta^{[M}{}_{J}\delta^{N]}{}_{K}=\frac{1}{2}r^I{}_{MN}\eps^{MNL}\eps_{JKL}=\frac{1}{2}{r}^{IL}\eps_{LJK},
\label{r-re}
\end{equation}
where
\begin{equation}
{r}^{IL}:=r^I{}_{MN}\eps^{MNL}.
\label{rIJ}
\end{equation}
On the other hand, the components $({r}^{IJ})$ admit the following decomposition: 
\begin{equation}
{r}^{IJ}={r}^{[IJ]}+{r}^{(IJ)}={r}^{[IJ]}+\ovcirc{r}^{IJ}+\frac{1}{3}{r}^{K}{}_{K}\delta^{IJ},
\label{r-decomp}
\end{equation}
where $(\ovcirc{r}^{IJ})$ is the traceless symmetric part of $({r}^{IJ})$:
\begin{align*}
\ovcirc{r}^{IJ}&:=r^{(IJ)}-\frac{1}{3}{r}^{K}{}_{K}\delta^{IJ}, & \ovcirc{r}^{IJ}\delta_{IJ}&=0.
\end{align*}

For a fixed point $y\in\Sigma$, we can distinguish a group of linear transformations of the tangent space $T_y\Sigma$, which preserve \emi the scalar product on $T_y\Sigma$ being the value of the metric $q$ at $y$ and \emii the orientation of the tangent space induced by the orientation of $\Sigma$. Since in the basis of $T_y\Sigma$ induced by the coframe $(\theta^I)$ the components of the scalar product are $(\delta_{IJ})$, the group can be described by $SO(3)$ matrices. Let us consider the following representation of $SO(3)$ on tensors of type $\binom{2}{0}$:   
\[
r^{IJ}(y)\mapsto \Lambda^{I}{}_K\Lambda^J{}_L\,r^{KL}(y),
\]
where $(\Lambda^I{}_J)\in SO(3)$. It is well known that Equation \eqref{r-decomp} describes a decomposition of this representation into irreducible ones.

We can also introduce a similar (and equivalent) representation of $SO(3)$ on those tensors of type $\binom{1}{2}$, which are antisymmetric in the two lower indices. A decomposition of this representation into irreducible ones can be obtained from \eqref{r-decomp} by means of \eqref{r-re}: 
\[
r^{I}{}_{JK}=\frac{1}{2}{r}^{[IL]}\eps_{LJK}+\frac{1}{2}\ovcirc{r}^{IL}\eps_{LJK}+\frac{1}{6}{r}^{K}{}_{K}\eps^I{}_{JK}.
\]
We will call the terms at the r.h.s. above antisymmetric, traceless symmetric and trace parts of $(r_I)$.

\paragraph{The derivatives $(d\theta^I)$} Similarly,
\begin{equation}
d\theta^I=\beta^{I}{}_{JK}\,\theta^J\ot\theta^K=\frac{1}{2}\beta^I{}_{JK}\,\theta^J\we\theta^K,
\label{dthI-comp}
\end{equation}
where $\beta^I{}_{JK}=\beta^I{}_{[JK]}$. Therefore
\begin{equation}
\beta^I{}_{JK}=\frac{1}{2}\beta^{IL}\eps_{LJK},
\label{b-be}
\end{equation}
where
\begin{equation}
{\beta}^{IL}:=\beta^I{}_{MN}\eps^{MNL}
\label{betaIJ}
\end{equation}
and
\begin{equation}
{\beta}^{IJ}={\beta}^{[IJ]}+\ovcirc{\beta}^{IJ}+\frac{1}{3}{\beta}^{K}{}_{K}\delta^{IJ},
\label{b-decomp}
\end{equation}
where $(\ovcirc{\beta}^{IJ})$ is the traceless symmetric part of $({\beta}^{IJ})$. The equation above and Equation \eqref{b-be} define the antisymmetric, traceless symmetric and trace parts of $(d\theta^I)$ like in the case of $(r_I)$.

\paragraph{Derivatives in the coframe $(\theta^I)$} Since we are going to carry out many calculations using components of the fields in the coframe $(\theta^I)$, it will be convenient to use in these calculations derivatives of the components with respect to vector fields, which constitute the dual frame. Let then $(\varepsilon_I)$ be the frame dual to $(\theta^I)$, that is, let $\varepsilon_I\lr\theta^J=\delta^J{}_I$. The derivative of a function $f$ with respect to $\varepsilon_I$ will be denoted by $f_{,I}$:        
\begin{equation}
f_{,I}\equiv \varepsilon_If.
\label{f,I}
\end{equation}
Then for any function $f$
\begin{equation}
df=f_{,I}\theta^I.
\label{df}
\end{equation}
Let us note that by virtue of \eqref{tt-q}
\begin{equation}
\varepsilon_I=\vec{\theta}_I,
\label{eps-vecth}
\end{equation}
provided $\xi^I=0$.

\subsection{The boost constraint}

Using \eqref{eps-vecth}, \eqref{zeta-th} and \eqref{zeta-ijkl} we can express the first term of the boost constraint \eqref{boost-0} as follows:
\[
\vec{\theta}^I\lr \zeta_I=\frac{1}{2\cdot3!}\bar{\zeta}^I\eps_{IKL}\theta^K\we\theta^L=\frac{1}{2\cdot3!}\bar{\zeta}_I\eps^{IKL}\theta_K\we\theta_L.
\]
On the other hand, the second term there
\begin{multline}
\theta_K\we{\hd}d\theta^K=\frac{1}{2}\beta^{K}{}_{MN}\theta_K\we {\hd}(\theta^M\we\theta^N)=\frac{1}{2}\beta^{K}{}_{MN}\eps^{MN}{}_L\theta_K\we \theta^L=\frac{1}{2}\beta^{KL}\theta_K\we \theta_L=\\=\frac{1}{2}\beta^{[KL]}\theta_K\we \theta_L
\label{tdt-Ab}
\end{multline}
---here in the second step we applied the second equation in \eqref{star-thetas}, and in the third one Equation \eqref{betaIJ}. Consequently,
\[
\frac{1}{3!}\bar{\zeta}_I\eps^{IKL}+\beta^{[KL]}=0.
\]
This equation by means of the last equation in \eqref{ee} can be transformed into
\begin{equation}
\bar{\zeta}_M+3\,\eps_{MKL}\beta^{[KL]}=0.
\label{boost-sol}
\end{equation}

We can interpret the result as a general solution to the boost constraint: it fixes the momenta $(\zeta_I)$ as a linear function of the antisymmetric part of $(d\theta^I)$ and does not impose any restriction on $\theta^I$. 

\subsection{The rotation constraint}

Using \eqref{eps-vecth}, \eqref{rI-comp} and \eqref{r-re} we obtain from \eqref{rot-0}
\[
  \vec{\theta}^I\lr r_I=r^I{}_{IN}\theta^N=\frac{1}{2}r^{IJ}\eps_{JIN}\theta^N=0.
\]
Thus
\[
r^{IJ}\eps_{JIN}=r^{[IJ]}\eps_{JIN}=0,
\]
which by virtue of the second equation in \eqref{ee} is equivalent to
\begin{equation}
r^{[IJ]}=0.
\label{Ar=0}
\end{equation}

Thus we obtained a general solution to the rotation constraint: it forces the antisymmetric part of $(r_I)$ to vanish. 

\subsection{The vector constraint in terms of $r^{IJ} $ and $\beta^{IJ}$}

Since we solved both boost and rotation constraints, it is enough to restrict our attention to the vector and the scalar constraints. Now we are going to express the constraints in terms of the variables $(r^{IJ})$ and $(\beta^{IJ})$, keeping in mind the result \eqref{Ar=0}.

Let us consider first the vector constraint \eqref{vec-0}---it will be convenient to act on both sides of the constraint by the Hodge operator ${\hd}$: 
\[
{\hd}dr^J+{\hd}(d\theta_I\we \vec{\theta}^J\lr r^I)=0.
\]
Combining \eqref{rI-comp} and \eqref{r-re} we get
\[
r^J=\frac{1}{4}r^{JM}\eps_{MKL}\theta^K\we\theta^L.
\]
Therefore
\begin{multline*}
dr^J=\frac{1}{4}r^{JM}{}_{,N}\,\eps_{MKL}\,\theta^N\we\theta^K\we\theta^L+\frac{1}{2}r^{JM}\eps_{MKL}\,d\theta^K\we\theta^L=\\=\frac{1}{4}r^{JM}{}_{,N}\,\eps_{MKL}\,\theta^N\we\theta^K\we\theta^L+\frac{1}{8}r^{JM}\eps_{MKL}\beta^{KN}\eps_{NPR}\,\theta^P\we\theta^R\we\theta^L,
\end{multline*}
where in the last step we used \eqref{dthI-comp} and \eqref{b-be}. By virtue of \eqref{star-thetas} and \eqref{ee}\[
{\hd}dr^J=\frac{1}{2}r^{JM}{}_{,M}+\frac{1}{4}r^{JM}\eps_{MKL}\beta^{[KL]}.
\]
On the other hand, applying Equations \eqref{tt-q}, \eqref{rI-comp} and \eqref{dthI-comp} we obtain
\[
d\theta_I\we\vec{\theta}^J\lr r^I=\frac{1}{8}\beta_{I}{}^P\eps_{PKL}\,r^{IM}\eps_{M}{}^{J}{}_{N}\,\theta^{K}\we\theta^L\we\theta^N,
\]
hence
\[
{\hd}(d\theta_I\we\vec{\theta^J}\lr r^I)=\frac{1}{4}\beta_I{}^Nr^{IM}\eps_{M}{}^{J}{}_{N}=\frac{1}{4}\beta_I{}^Nr^{IM}\eps^J{}_{NM},
\]
where we used \eqref{star-thetas}. Thus the vector constraint can be expressed in the following form: 
\[
{2}r^{JM}{}_{,M}+r^{JM}\beta^{[KL]}\eps_{MKL}+r^{IM}\beta_I{}^N\eps^J{}_{NM}=0.
\]
In terms of the decompositions \eqref{r-decomp} and \eqref{b-decomp} the constraint reads
\begin{equation}
\frac{2}{3}r^K{}_K{}^{,J}+2\ovcirc{r}^{JM}{}_{,M}+\ovcirc{r}^{JM}\beta^{[KL]}\eps_{MKL}+\ovcirc{r}^{IM}\beta_{[IN]}\eps^{JN}{}_M+\ovcirc{r}^{IM}\ovcirc{\beta}_{IN}\eps^{JN}{}_M=0.
\label{vec-cmp}
\end{equation}

\subsection{The scalar constraint in terms of $r^{IJ} $ and $\beta^{IJ}$}

Let us now turn to the scalar constraint \eqref{*scal}. The zero-form ${\hd}(r_I\we\theta^J)$ appearing in the first two terms of the constraint can be transformed as follows:
\[
{\hd}(r_I\we\theta^J)=\frac{1}{2}r_{IMN}{\hd}(\theta^M\we\theta^N\we\theta^J)=\frac{1}{2}r_{IMN}\eps^{MNJ}=\frac{1}{2}r_I{}^J
\] 
---here in the first step we used \eqref{rI-comp}, in the second one \eqref{star-thetas} and finally in the third one \eqref{rIJ}. Thus
\begin{multline*}
[{\hd}(r_I\we\theta^J)][{\hd}(r_J\we\theta^I)]-\frac{1}{2}[{\hd}(r_I\we\theta^I)]^2=\frac{1}{4}\Big[r_{I}{}^Jr_J{}^I-\frac{1}{2}(r_I{}^I)^2\Big]=\frac{1}{4}\Big[r^{IJ}r_{JI}-\frac{1}{2}(r^K{}_K)^2\Big]=\\=\frac{1}{4}\Big[\Big(r^{[IJ]}+\ovcirc{r}^{IJ}+\frac{1}{3}r^K{}_K\delta^{IJ}\Big)\Big(r_{[JI]}+\ovcirc{r}_{JI}+\frac{1}{3}r^K{}_K\delta_{JI}\Big)-\frac{1}{2}(r^K{}_K)^2\Big]=\\=\frac{1}{4}\Big[-r^{[IJ]}r_{[IJ]}+\ovcirc{r}^{IJ}\ovcirc{r}_{IJ}-\frac{1}{6}(r^K{}_K)^2\Big]. 
\end{multline*}
Note that at every point of the manifold $\Sigma$, the formula above defines a (nondegenerate) quadratic form
\[
(r^{IJ})\mapsto Q(r^{IJ})
\]
of signature $(5,4,0)$. Taking into account the rotation constraint \eqref{Ar=0}, we see that this form reduces to a quadratic form of signature $(5,1,0)$ defined on the symmetric part of $(r^{IJ})$ .  

An analogous calculation leads to a similar expression for the last two terms in \eqref{*scal}:
\[
[{\hd}(d\theta_I\we\theta^J)][{\hd}(d\theta_J\we\theta^I)]-\frac{1}{2}[{\hd}(d\theta_I\we\theta^I)]^2=\frac{1}{4}\Big[-\beta^{[IJ]}\beta_{[IJ]}+\ovcirc{\beta}^{IJ}\ovcirc{\beta}_{IJ}-\frac{1}{6}(\beta^K{}_K)^2\Big].
\]

Finally, the middle term in \eqref{*scal}
\begin{multline}
2{\hd}d(\theta_J\we{\hd}d\theta^J)={\hd}d(\beta^{[KL]}\theta_K\we\theta_L)={\hd}(\beta^{[KL],M}\theta_M\we\theta_K\we\theta_L+\\+\beta^{[KL]}\beta_{KMN}\,\theta^M\we\theta^N\we\theta_L)=(\beta^{[KL],M}\eps_{KLM}+\beta^{[KL]}\beta_{KMN}\eps^{MN}{}_L)=\\=(\beta^{[KL],M}\eps_{KLM}+\beta^{[KL]}\beta_{[KL]}).
\label{beta'}
\end{multline}

Thus the scalar constraint reads
\begin{equation}
\ovcirc{r}^{IJ}\ovcirc{r}_{IJ}-\frac{1}{6}(r^K{}_K)^2+{4}\beta^{[KL],M}\eps_{KLM}+3\beta^{[IJ]}\beta_{[IJ]}+\ovcirc{\beta}^{IJ}\ovcirc{\beta}_{IJ}-\frac{1}{6}(\beta^K{}_K)^2=0
\label{scal-cmp}
\end{equation}
and it has a transparent algebraic form: the constraint is a sum of a term linear in derivatives of $\beta^{[IJ]}$ and quadratic forms of $r^I$ and $d\theta^I$. Both quadratic forms are diagonal when expressed in terms of the components of the irreducible parts of the variables.

\section{Further simplification of the constraints? \label{gauge?}}

Since the boost and rotation constraints have been solved, we are left with the vector and the scalar constraints
\eqref{vec-cmp} and \eqref{scal-cmp}. Can this two constraints be simplified even further? Let us note that the only second order term in the constraints is that in the scalar constraint, which is linear in derivatives of $\beta^{[IJ]}$ and it would be nice to be able to reduce the order of the constraints by annihilating this term.  Perhaps this can be done by means of the remaining gauge freedom. Of course, we would like to achieve this {\em without spoiling the gauge} $\xi^I=0$.

The term, we would like to annihilate, originates from the three-form $d(\theta_I\we\hd d\theta^I)$ (see Equations \eqref{tdt-Ab} and \eqref{beta'}). Therefore it seems reasonable to check, if either  $d(\theta_I\we\hd d\theta^I)$ or $\theta_I\we\hd d\theta^I$ can be gauge-transformed to zero. Below we will focus on the latter alternative because of two reasons. First, this alternative leads to first order PDE's imposed on parameters of gauge transformations, while the former one leads to second order equations. Second, by virtue of \eqref{tdt-Ab} the latter alternative is equivalent to    
\begin{equation}
\beta^{[KL]}=0,
\label{Ab=0}
\end{equation}
which would lead to a greater simplification of the constraints \eqref{vec-cmp} and \eqref{scal-cmp} than the former alternative (compare the l.h.s of \eqref{Ab=0} with the r.h.s. of \eqref{beta'}).

The question, we would like to answer, can be now formulated precisely as follows: given any point $(\zeta^0_I,r^0_J,\xi^K_0=0,\theta^L_0)$ of the phase space, is it possible to gauge-transform it to a point, where the antisymmetric part of $d\theta^L$ is zero, keeping at the same time the gauge $\xi^K=0$ intact? Currently we are unable to provide a definite solution to this problem and therefore we will present below merely some partial results.             

\subsection{$SO(3)$ gauge transformations}

Since the boost constraint \eqref{boost-c} acts non-trivially on $\xi^I$, it would be rather difficult to use it to annihilate $\beta^{[IJ]}$ without spoiling the gauge $\xi^I=0$. But the rotation constraint \eqref{rot-c} does not depend on the momenta $(\zeta_I)$ conjugate to $(\xi^I)$ and therefore gauge transformations generated by the constraint do not change the functions $(\xi^I)$. On the other hand, the constraint generates local $SO(3)$ transformations of the (orthonormal) coframe $(\theta^I)$ \cite{oko-tegr-I}. It is thus reasonable to check if the condition $\beta^{[IJ]}=0$ can be achieved by these transformations applied to an arbitrary coframe $(\theta^I)$. 

To avoid any confusion let as emphasize that the components $(\beta^{[IJ]})$ describe an irreducible subspace of the $SO(3)$ representation acting on the components of $(d\theta^I)$. Obviously, this representation cannot be used to annihilate the antisymmetric part of $(d\theta^I)$. Therefore below we will consider the action of $SO(3)$ on $(\theta^I)$ generated by the rotation constraint, which will give us an affine transformation of $(\beta^{[IJ]})$.         

To begin the analysis of the problem, let us assume that the gauge $\xi^I=0$ holds and let us describe the gauge transformations of the coframe $(\theta^I)$ generated by the rotation constraint. To this end consider a map
\begin{equation}
\R\supset\,]\lambda_1,\lambda_2[\,\ni\lambda\to b(\lambda)\in\Omega^1({\Sigma})
\label{l-bl}
\end{equation}
such that for every $\la$ the one-form $b(\lambda)$ is independent of the canonical variables. The transformations are given by the following differential equations (treated pointwisely with respect to the manifold $\Sigma$):
\[
\frac{d\theta^I}{d\lambda}=\{\theta^I,R(b(\lambda))\}=\frac{\delta R(b(\lambda))}{\delta r_I}={\hd}\big(b(\lambda)\we\theta^I).
\]
Decomposing $b(\lambda)$ in the coframe $(\theta^I)$ and using \eqref{star-thetas} we obtain
\begin{equation}
\frac{d\theta^I}{d\lambda}={\hd}\big(b_J(\lambda)\,\theta^J\we\theta^I)=b_J(\lambda)\,\eps^{JI}{}_K\theta^K\equiv \Ec^I{}_K(\lambda)\,\theta^K.
\label{R-theta}
\end{equation}

We thus see that the gauge transformations are given by a system of linear differential equations imposed on values of the variables $(\theta^I)$ at every point of $\Sigma$. Let $\lambda\mapsto \tilde{\theta}^I(\lambda)$ be the solution of \eqref{R-theta} with the initial condition $\tilde{\theta}(0)=\theta^I_0$. Then    
\begin{equation}
\tilde{\theta}^I(\lambda)=\Lambda^I{}_J(\lambda)\,\theta^J_0,
\label{t-Lt}
\end{equation}
where $\big(\Lambda^I{}_J(\lambda)\big)$ is a field on $\Sigma$ valued in $3\times 3$ real matrices. If $\Ec(\lambda)$ denotes a linear operator on $\mathbb{E}$ of components $\big(\Ec^I{}_J(\la)\big)$, then
\begin{equation}
\Lambda^I{}_J(\lambda)={\Big[\Exp\Big(\int_{0}^\lambda \Ec(\lambda')\,d\lambda'\Big)\Big]^I}_J,
\label{Exp}
\end{equation}
where $\Exp$ denotes the ordered exponential. 

Let us consider now a Riemannian metric given by the transformed coframe $(\tilde{\theta}^I(\la))$:
\begin{equation}
\tilde{q}(\la):=\delta_{IJ}\,\tilde{\theta}^I(\la)\ot\tilde{\theta}^J(\la)=\delta_{IJ}\Lambda^I{}_K(\la)\Lambda^J{}_L(\la)\,\theta^K_0\ot\theta^L_0.
\label{q-la}
\end{equation}
Using \eqref{R-theta} we obtain
\[
\frac{d \tilde{q}}{d\la}=\delta_{IJ}\Big(\frac{d\tilde{\theta}^I}{d\lambda}\ot\tilde{\theta}^J+\tilde{\theta}^I\ot \frac{d\tilde{\theta}^J}{d\lambda}\Big)=\Ec_{JK}\tilde{\theta}^K\ot\tilde{\theta}^J+\Ec_{IK}\tilde{\theta}^I\ot\tilde{\theta}^K=\Ec_{JK}(\tilde{\theta}^K\ot\tilde{\theta}^J+\tilde{\theta}^J\ot\tilde{\theta}^K)=0,
\]
since $\Ec_{JK}=-\Ec_{KJ}$. Thus the map $\la\to\tilde{q}(\lambda)$ is constant and therefore $\tilde{q}(\la)=\tilde{q}(0)$. This conclusion can be reexpressed by means of  \eqref{q-la} in the following way: 
\[
\delta_{IJ}\Lambda^I{}_K(\la)\Lambda^J{}_L(\la)\,\theta^K_0\ot\theta^L_0=\delta_{KL}\,\theta^K_0\ot\theta^L_0,
\]
which means that
\begin{equation}
\delta_{IJ}\Lambda^I{}_K(\la)\Lambda^J{}_L(\la)=\delta_{KL}.
\label{dLL=d}
\end{equation}
Consequently, for every $\la$ the field $\big(\Lambda^I{}_K(\la)\big)$ is valued in $SO(3)$ matrices. This justifies the statement that the rotation constraint \eqref{rot-c} generates local $SO(3)$ transformations of the coframe $(\theta^I)$, where ``local'' means ``dependent on a point of $\Sigma$''.     

\subsection{Gauge transformation of antisymmetric part of $(d\theta^I_0)$}

Now we are going to derive a formula, which describes the corresponding gauge transformation of the antisymmetric part of $(d\theta^I_0)$. Let us then simplify the notation:
\begin{align*}
\tilde{\theta}^I(\la)&\equiv\tilde{\theta}^I, & \Lambda^I{}_{J}(\la)&\equiv \Lambda^I{}_{J}
\end{align*}
and consider the following terms:
\begin{align}
&{\hd}(\tilde{\theta}_{K}\we{\hd}d\tilde{\theta}^K)=\frac{1}{2}\tilde{\beta}^{[KL]}\eps_{KLM}\,\tilde{\theta}^M\equiv\tilde{\beta}^M\,\tilde{\theta}_{M},\label{tb_S}\\
  &{\hd}(\theta_{0K}\we{\hd}d\theta^K_0)=\frac{1}{2}\beta^{[KL]}\eps_{KLM}\,\theta_0^M\equiv\beta^M\,\theta_{0M},
\label{b0_S}
\end{align}
---these two equations were obtained by means of \eqref{tdt-Ab} and \eqref{star-thetas}. Inserting \eqref{t-Lt} into \eqref{tb_S} we get
\begin{multline}
\tilde{\beta}^M\,\tilde{\theta}_{M}={\hd}(\delta_{IJ}\tilde{\theta}^I\we{\hd}d\tilde{\theta}^J)=\delta_{IJ}{\hd}\big(\Lambda^I{}_{M}{\theta}^M_0\we{\hd}d(\Lambda^J{}_N{\theta}_0^N)\big)=\\=\delta_{IJ}{\hd}\big(\Lambda^I{}_{M}{\theta}_0^M\we{\hd}(\Lambda^J{}_{N,L}\theta^L_0\we{\theta}_0^N+\Lambda^J{}_Nd\theta_0^N)\big)=\\=\delta_{IJ}\Lambda^I{}_{M}\Lambda^J{}_{N,L}{\hd}\big({\theta}_0^M\we{\hd}(\theta_0^L\we{\theta}_0^N)\big)+\delta_{IJ}\Lambda^I{}_{M}\Lambda^J{}_{N}{\hd}\big({\theta}_0^M\we{\hd}d\theta_0^N\big)=\\=\delta_{IJ}\Lambda^I{}_{M}\Lambda^J{}_{N,L}\eps^{LN}{}_K\eps^{MK}{}_P\,\theta_0^P+\beta^P\,\theta_{0P}
\label{LL}
\end{multline}
---here in the second step we applied \eqref{df} and in the last step we used twice \eqref{star-thetas}, the condition \eqref{dLL=d} and Equation \eqref{b0_S}. Equation \eqref{t-Lt} allows us to express \eqref{LL} as follows:
\begin{equation}
\tilde{\beta}^M\Lambda_{M}{}^P=\delta_{IJ}\Lambda^I{}_{M}\Lambda^J{}_{N,L}\eps^{LN}{}_K\eps^{MKP}+\beta^P.
\label{aff}
\end{equation}
We thus see that the variables $\tilde{\beta}^M$ and $\beta^M$, and thereby the antisymmetric parts of $(d\tilde{\theta}^I)$ and $(d\theta^I_0)$, are related by an affine transformation, as mentioned earlier.    

Suppose now that for a certain $\la_1$ the antisymmetric part of $(d\tilde{\theta}^I)$ is zero everywhere on $\Sigma$. Then Equation \eqref{aff} transforms into the following system of quadratic PDE's imposed on components of the $SO(3)$ matrix $(\Lambda^I{}_J)\equiv\big(\Lambda^I{}_J(\la_1)\big)$:
\begin{multline}
0=\delta_{IJ}\Lambda^I{}_{M}\Lambda^J{}_{N,L}\eps^{LN}{}_K\eps^{MKP}+\beta^P=\\=-\delta_{IJ}\Lambda^I{}_{M}\Lambda^J{}_{N,L}\delta^{LM}\delta^{NP}+\delta_{IJ}\Lambda^I{}_{M}\Lambda^J{}_{N,L}\delta^{LP}\delta^{NM}+\beta^P
\label{qpde}
\end{multline}
---the second equality here holds by virtue of the second identity in \eqref{ee}. 

The system above can be easily reduced to a system of linear PDE's. To this end let us note that Equation \eqref{dLL=d} implies the following identity:
\[
3=\delta_{IJ}\Lambda^I{}_{M}\Lambda^J{}_{N}\delta^{MN},
\]
hence
\[
0=(\delta_{IJ}\Lambda^I{}_{M}\Lambda^J{}_{N}\delta^{MN})_{,L}=2\delta_{IJ}\Lambda^I{}_{M}\Lambda^J{}_{N,L}\delta^{MN}.
\]
Thus the second term at the r.h.s. of \eqref{qpde} vanishes. On the other hand, differentiating \eqref{dLL=d}, we obtain
\[
\delta_{IJ}\Lambda^I{}_{M,L}\Lambda^J{}_{N}+\delta_{IJ}\Lambda^I_{M}\Lambda^J{}_{N,L}=0.
\]
Using this result we rewrite \eqref{qpde} as follows
\begin{equation}
\Lambda^I{}_{M,L}\delta^{LM}(\delta_{IJ}\Lambda^J{}_{N}\delta^{NP})+\beta^P=0.
\label{qpde-1}
\end{equation}
By virtue of \eqref{dLL=d} for each $SO(3)$ matrix
\[
\delta_{IJ}\Lambda^J{}_{N}\delta^{NP}=\Lambda^{-1P}{}_I.
\]
Therefore \eqref{qpde-1} can be transformed to
\begin{equation}
\Lambda^K{}_{M,L}\delta^{ML}+\Lambda^K{}_I\beta^I=0,
\label{lin}
\end{equation}
which is the desired system of linear PDE's equivalent to \eqref{qpde}. Note that we obtained a system of three decoupled equations labeled by the upper index $K$---each equation is imposed on three components $\Lambda^K{}_1$, $\Lambda^K{}_2$ and $\Lambda^K{}_3$ with the same $K$.

Thus in order to achieve \eqref{Ab=0} by means of the gauge transformation, we need to show that for every coframe $(\theta^I_0)$ there exists a curve \eqref{l-bl} such that for some $\la_1$ the corresponding matrix-valued field  $\big(\Lambda^I{}_J(\la_1)\big)$ satisfies \eqref{lin} (the curve defines the field via Equations \eqref{Exp} and \eqref{R-theta}).   

Since we are not able to show this, we will consider a similar but simpler problem: given Riemannian metric $q$ on $\Sigma$, does there exist a (local) orthonormal coframe of $q$, for which $\beta^{[IJ]}=0$? In this case, too, we are unable to give a definite answer, but at least we will able to take the analysis of the problem a few steps further.

To make this task easier, given metric $q$, we will choose a special coframe $(\theta^I_0)$ orthonormal w.r.t. $q$, distinguished by its simplicity. Then will look for an $SO(3)$-valued  field $(\Lambda^I{}_J)$ such that the antisymmetric part of $\big(d(\Lambda^I{}_J\theta^J_0)\big)$ vanishes. In other words, we will look for a field $(\Lambda^I{}_J)$, which satisfies \eqref{dLL=d} and \eqref{lin} with $(\beta^I)$ given by that simple coframe.

\subsection{Darboux coframe}

The Darboux theorem \cite{darboux} guarantees that every Riemannian metric $q$ defined on a {\em three-dimensional} manifold can always be expressed in a diagonal form. More precisely, around every point of such a manifold there exists a local coordinate system $(x,y,z)$ such that
\begin{equation}
q=A^2dx^2+B^2dy^2+C^2dz^2,
\label{q-dar}
\end{equation}
where $A$, $B$ and $C$ are functions of non-zero values, defined on the domain of the coordinate system. Consequently, 
\begin{align}
\theta^1&=Adx, & \theta^2&=Bdy, & \theta^3&=Cdz,
\label{ansatz-1}
\end{align}
is an orthonormal coframe of $q$. Any (orthonormal) coframe of this form will be called {\em Darboux coframe} ({\em of} $q$).

To analyze Equation \eqref{lin}, we will assume that the coframe $(\theta^I_0)$, which defines $(\beta^I)$ in the equation, is a Darboux coframe of $q$. Its dual frame, which defines the derivatives \eqref{f,I} in Equation \eqref{lin}, reads
\begin{align}
\varepsilon_1&=\frac{1}{A}\partial_x, & \varepsilon_2&=\frac{1}{B}\partial_y, & \varepsilon_3&=\frac{1}{C}\partial_z.
\label{dual-eps}
\end{align}

Now it is easy to calculate the exterior derivatives of the one-forms \eqref{ansatz-1}:
\begin{align*}
d\theta^1&=-\frac{A_{,y}}{AB}\theta^1\we\theta^2+\frac{A_{,z}}{CA}\theta^3\we\theta^1,\\
d\theta^2&=\frac{B_{,x}}{AB}\theta^1\we\theta^2-\frac{B_{,z}}{BC}\theta^2\we\theta^3,\\
d\theta^3&=-\frac{C_{,x}}{CA}\theta^3\we\theta^1+\frac{C_{,y}}{BC}\theta^2\we\theta^3.
\end{align*}
Taking into account Equations \eqref{betaIJ} and \eqref{e-sgn}, it is convenient to denote
\[
\sgn(\theta^K)\beta^{IJ}\equiv\bar{\beta}^{IJ}.
\]
It follows from the expressions for $d\theta^I$ above that 
\begin{equation}
\Big(\bar{\beta}^{IJ}\Big)=2
\begin{pmatrix}
0 & \dfrac{A_{,z}}{AC} & -\dfrac{A_{,y}}{AB} \medskip\\
-\dfrac{B_{,z}}{BC} & 0 &  \dfrac{B_{,x}}{BA}\medskip\\
\dfrac{C_{,y}}{CB} & -\dfrac{C_{,x}}{CA} & 0
\end{pmatrix}=
2\begin{pmatrix}
0 & \dfrac{(\ln|A|)_{,z}}{C} & -\dfrac{(\ln|A|)_{,y}}{B} \medskip\\
-\dfrac{(\ln|B|)_{,z}}{C} & 0 &  \dfrac{(\ln|B|)_{,x}}{A}\medskip\\
\dfrac{(\ln|C|)_{,y}}{B} & -\dfrac{(\ln|C|)_{,x}}{A} & 0
\end{pmatrix}
\label{b-matr}
\end{equation}
Consequently,
\begin{align}
\bar{\beta}^{[12]}&=\frac{(\ln|A|+\ln|B|)_{,z}}{C}, & \bar{\beta}^{[23]}&=\frac{(\ln|B|+\ln|C|)_{,x}}{A}, & \bar{\beta}^{[31]}&=\frac{(\ln|C|+\ln|A|)_{,y}}{B}.\label{Ab}
\end{align}
and
\begin{align}
\bar{\beta}^1&=\frac{(\ln|BC|)_{,x}}{A}, & \bar{\beta}^2&=\frac{(\ln|CA|)_{,y}}{B}, & \bar{\beta}^3&=\frac{(\ln|AB|)_{,z}}{C}
\label{bbar}
\end{align}
(see \eqref{b0_S} for the relation between $\beta^{[IJ]}$ and $\beta^K$).

To simplify the considerations below let us note that Equation \eqref{lin} is invariant with respect to the transformation $(\theta^I_0)\mapsto (-\theta^I_0)$. Indeed, the dual frame changes then according to $(\varepsilon_I)\mapsto (-\varepsilon_I)$, which results in the change in the sign of the first term in \eqref{lin}. On the other hand, the transformation of the coframe changes the sign of the components $(\beta^I{}_{JK})$ and $(\eps_{IJK})$---see Equations \eqref{dthI-comp} and \eqref{e-sgn}. This in turn, by virtue of Equations \eqref{betaIJ} and \eqref{b0_S}, changes the sign of $(\beta^I)$ appearing in the second term in \eqref{lin}. Thus without loss of generality we can assume that $(\theta^I_0)$ is a Darboux coframe compatible with the orientation of $\Sigma$ and then drop the bar over $\beta^I$ in  \eqref{bbar}. 

\subsection{Equations \eqref{dLL=d} and \eqref{lin} in the case of the Darboux coframe}

Now, inserting \eqref{dual-eps} and \eqref{bbar} into \eqref{lin} we obtain: 
\begin{equation}
\frac{\Lambda^K{}_{1,x}}{A}+\frac{\Lambda^K{}_{2,y}}{B}+\frac{\Lambda^K{}_{3,z}}{C}+\Lambda^K{}_1\frac{(\ln|BC|)_{,x}}{A}+\Lambda^K{}_2\frac{(\ln|CA|)_{,y}}{B}+ \Lambda^K{}_3 \frac{(\ln|AB|)_{,z}}{C}=0.
\label{lin-expl}
\end{equation}
Multiplying both sides of the equation above by $ABC$ gives us
\begin{equation}
(\Lambda^K{}_1BC)_{,x}+(\Lambda^K{}_2CA)_{,y}+(\Lambda^K{}_3AB)_{,z}=0.
\label{div-b}
\end{equation}
To proceed further let us introduce a non-singular matrix
\[
(\Phi^I{}_J):={\rm diag}(BC,CA,AB)
\]
and let 
\[
\bar{\Lambda}^I{}_J:=\Lambda^I{}_K\Phi^K{}_J.
\]
Then Equations \eqref{dLL=d} and \eqref{div-b} are equivalent to, respectively,
\begin{align}
\bar{\Lambda}^I{}_M\bar{\Lambda}^J{}_N\delta_{IJ}&=\delta_{KL}\Phi^{K}{}_M\Phi^L{}_N, \label{LL-FF}\\
\bar{\Lambda}^K{}_{1,x}+\bar{\Lambda}^K{}_{2,y}+\bar{\Lambda}^K{}_{3,z}&=0.\label{dbarL}
\end{align}

The last equation can be easily solved: if
\[
\Omega^K:=\bar{\Lambda}^K{}_{1}\,dy\we dz + \bar{\Lambda}^K{}_{2}\,dz\we dx + \bar{\Lambda}^K{}_{3}\,dx\we dy,
\]
then Equation \eqref{dbarL} can be rewritten as
\[
d\Omega^K=0.
\]
By virtue of the Poincar\'e lemma this is equivalent to the existence of a one-form
\[
\omega^K=\omega^K{}_1\,dx+\omega^K{}_2\,dy+\omega^K{}_3\,dz
\]
(defined on a simply connected open subset of $\Sigma$) such that 
\[
\Omega^K=d\omega^K.
\]
This fact gives us a general solution to \eqref{dbarL}
\begin{align}
\bar{\Lambda}^K{}_1&=\omega^K{}_{3,y}-\omega^K{}_{2,z},&
\bar{\Lambda}^K{}_2&=\omega^K{}_{1,z}-\omega^K{}_{3,x},&
\bar{\Lambda}^K{}_3&=\omega^K{}_{2,x}-\omega^K{}_{1,y}.
\label{bL-om}
\end{align}
where $(\omega^K{}_L)$ are arbitrary functions.

Inserting $(\bar{\Lambda}^I{}_J)$ of the form above into \eqref{LL-FF} we obtain a system of quadratic first order PDE's imposed on the functions $(\omega^K{}_L)$---if there exist solutions to this system, then there exists an orthonormal coframe $(\tilde{\theta}^I)$ of $q$ such that  its exterior derivative $(d\tilde{\theta}^I)$ has a vanishing antisymmetric part.

To express the system in a more compact way let us introduce an auxiliary metric $\check{q}$ and its volume form $\check{\eps}$ on the domain of the coordinates $(x,y,z)\equiv(x^1,x^2,x^3)\equiv(x^I)$: 
\begin{align*}
\check{q}&:=dx^2+dy^2+dz^2, & \check{\eps}&:=dx\we dy\we dz.
\end{align*}
Then \eqref{bL-om} is equivalent  to
\[
\bar{\Lambda}^I{}_J\,dx^J=\check{\hd}\,d\omega^I,
\]
where $\check{\hd}$ denotes the Hodge dualization given by $\check{q}$ and $\check{\eps}$. Substituting this into Equation \eqref{LL-FF} contracted with $dx^M\ot dx^N$ gives an alternative form of the system under consideration:
\begin{multline}
\delta_{IJ}\,(\check{\hd}\,d\omega^I)\ot(\check{\hd}\,d\omega^J)=(BC)^2dx\ot dx + (CA)^2dy\ot dy + (AB)^2dz\ot dz=\\=q_{yy}q_{zz}\,dx\ot dx + q_{zz}q_{xx}\,dy\ot dy + q_{xx}q_{yy}\,dz\ot dz
\label{om2-ABC}
\end{multline}
---this is a tensor equation imposed on the three unknown one-forms $(\omega^I)$, where the functions $q_{xx}=A^2$, $q_{yy}=B^2$ and $q_{zz}=C^2$ are the diagonal components of the diagonal metric \eqref{q-dar}.

So far we have not been able to prove the existence of solutions to the equation above for arbitrary diagonal metric $q$. However, both sides of \eqref{om2-ABC} are symmetric (local) tensor fields on the three-dimensional manifold $\Sigma$ and therefore \eqref{om2-ABC} is in fact a system of six equations imposed on nine independent components $(\omega^K{}_L)$. This suggests that it is possible to solve this system, but obviously such a naive counting cannot be treated as a proof.

Note moreover that the diagonal components of the matrix \eqref{b-matr} are zero. This means that for every Riemannian metric $q$ on $\Sigma$ there exists a special (local) orthonormal coframe $(\theta^I)$, that is, a Darboux coframe, such that three of nine components of $(\beta^{IJ})$ vanish. Thus it is perhaps possible to find an other special orthonormal coframe of $q$, for which three independent off-diagonal components of $(\beta^{[IJ]})$ do vanish.

\section{Solutions to the constraints}

Since we managed to reduce the complexity of the TEGR constraints to some extent, we can attempt to look for exact solutions to them. Let us recall that the boost and rotation constraints have been solved (see Equations \eqref{boost-sol} and \eqref{Ar=0}) and we are left with the vector and scalar constraints. Obviously, the latter constraints are still too complex to find their general solutions and therefore we will present just a few very specific ones. 

To derive the solutions we will keep the gauge $\xi^I=0$ and assume moreover that
\begin{enumerate}
\item $(\theta^I)$ is a Darboux coframe, 
\item the antisymmetric part of $(d\theta^I)$ vanishes, 
\item the traceless symmetric part of the momentum $(r^I)$ is zero. 
\end{enumerate}
Then we will find all Darboux coframes, for which $\beta^{[IJ]}=0$. We will thus obtain a class of coframes of fairly simple exterior derivatives---it follows from \eqref{b-matr} that within this class merely three off-diagonal components of $\beta^{(IJ)}$ will be non-zero.

On the other hand, the assumption $\ovcirc{r}^{IJ}=0$ will reduce the vector constraint to a system of trivial PDE's of obvious general solution. 

The results obtained in the previous steps will allow us to transform the remaining scalar constraint into a non-linear first-order PDE imposed on three functions, which  parameterize the class of Darboux coframes described above. Finally, we will find some exact solutions to this equation.

In total, we will obtain three families of exact solutions to all the constraints. For all these solutions both the functions $(\xi^I)$ and their conjugate momenta $(\zeta_I)$ will be zero, while the momenta $(r_I)$ will be non-zero---thus each solution will be specified by a coframe $(\theta^I)$ and its conjugate momenta $(r_I)$.

\subsection{Darboux coframes and $\beta^{[IJ]}=0$ }

Suppose then that $(\theta^I)$ is a Darboux coframe \eqref{ansatz-1} such that the antisymmetric part of $(d\theta^I)$ is zero. The components of this part for this coframe are given by Equations \eqref{Ab}, hence we have

\[
(\ln|A|+\ln|B|),_{z}=(\ln|B|+\ln|C|),_{x}=(\ln|C|+\ln|A|),_{y}=0.
\]
This means that
\begin{align}
\ln|A|+\ln|B|&=-2\ln|F|, & \ln|B|+\ln|C|&=-2\ln|G|, & \ln|C|+\ln|A|&=-2\ln |H|,
\label{lnln}
\end{align}
where $F$, $G$ and $H$ are functions of non-zero values such that \emi $F$ does not depend on $z$, \emii $G$ does not on $x$ and \emiii $H$ does not on $y$:
\begin{align*}
F&=F(x,y), & G&=G(y,z), & H&=H(x,z). 
\end{align*}
Equations \eqref{lnln} form a system of linear equations for $\ln|A|$, $\ln|B|$ and $\ln|C|$, which can be easily solved:
\begin{align*}
\ln|A|&=-\ln|F|+\ln|G|-\ln|H|,\\
\ln|B|&=-\ln|F|-\ln|G|+\ln|H|,\\
\ln|C|&= \ln|F|-\ln|G|-\ln|H|.
\end{align*}
Consequently,
\begin{align}
A&=\pm\frac{G(y,z)}{F(x,y)H(x,z)}, & B&=\pm\frac{H(x,z)}{G(y,z)F(x,y)}, & C&=\pm\frac{F(x,y)}{H(x,z)G(y,z)}.
\label{ABC}
\end{align}

We thus see that each Darboux coframe, whose exterior derivative is lacking its antisymmetric part, is of the form \eqref{ansatz-1} with the functions $A$, $B$ and $C$ given by Equations \eqref{ABC}.    

Let us now describe the symmetric part of $(d\theta^I)$. By virtue of \eqref{b-matr}   
\begin{align}
\bar{\beta}^{(12)}&=\frac{(\ln|A|-\ln|B|)_{,z}}{C}, & \bar{\beta}^{(23)}&=\frac{(\ln|B|-\ln|C|)_{,x}}{A}, & \bar{\beta}^{(31)}&=\frac{(\ln|C|-\ln|A|)_{,y}}{B}\label{Sb}
\end{align}
and
\begin{equation}
\bar{\beta}^{(11)}=\bar{\beta}^{(22)}=\bar{\beta}^{(33)}=\bar{\beta}^{K}{}_K=0.
\label{diag=0}
\end{equation}
Inserting \eqref{ABC} into \eqref{Sb} we obtain
\begin{align}
\bar{\beta}^{(12)}&=\pm 2\frac{G_{,z}H-H_{,z}G}{F}, & \bar{\beta}^{(23)}&=\pm2\frac{H_{,x}F-F_{,x}H}{G}, & \bar{\beta}^{(31)}&=\pm2\frac{F_{,y}G-G_{,y}F}{H},
\label{beta-sym}
\end{align}
where $\bar{\beta}^{(12)}$ inherits the sign from the function $C$, $\bar{\beta}^{(23)}$ from $A$ and $\bar{\beta}^{(31)}$ from $B$.

\subsection{Solving the vector constraint}

If $\ovcirc{r}^{IJ}=0$, then the vector constraint \eqref{vec-cmp} reduces to  
\[
r^K{}_K{}^{,J}=0.
\]
Thus the derivatives of the trace $r^K{}_K$ along all the vector fields \eqref{dual-eps} are zero. This means that the trace $r^K{}_K$ is constant on $\Sigma$:  
\begin{equation}
r^K{}_K\equiv 4\sqrt{3}\kappa= {\rm const.}
\label{r-tr}
\end{equation}

\subsection{Solving the scalar constraint}

Taking into account \emi the assumptions $\beta^{[IJ]}=0=\ovcirc{r}^{IJ}$ and \emii Equations \eqref{diag=0}, \eqref{beta-sym} and \eqref{r-tr}, we see that the scalar constraint \eqref{scal-cmp} takes the following form:
\begin{multline}
\frac{1}{4}\Big[\big(\beta^{(12)}\big)^2+\big(\beta^{(23)}\big)^2+\big(\beta^{(31)}\big)^2\Big]=\\=\frac{(G_{,z}H-H_{,z}G)^2}{F^2}+\frac{(H_{,x}F-F_{,x}H)^2}{G^2}+\frac{(F_{,y}G-G_{,y}F)^2}{H^2}=\kappa^2.
\label{pde-k2}
\end{multline}

Before we will present some solutions to the constraint \eqref{pde-k2}, let us first exclude those, which correspond in a quite simple way to the Minkowski spacetime. Note that if $\kappa=0$, then \emi \eqref{pde-k2} forces the only non-zero components of $(d\theta^I)$, that is, the off-diagonal components of $\beta^{(IJ)}$ to vanish and \emii \eqref{r-tr} forces the only non-zero part of the momenta $(r_I)$ to vanish. Consequently, $d\theta^I=0$ and $r_I=0$. The orthonormal frame $(\theta^I)$ is thus holonomic and the spatial metric $q$ is flat. On the other hand, if $\beta^{[IJ]}=0$, then the momenta $(\zeta_I)$ are zero by virtue of \eqref{boost-sol}. Obviously, the functions $(\xi^I)$ are still zero. It is easy to check that these canonical variables can be obtained on any spatial three-dimensional hyperplane in the Minkowski spacetime.

An example of a non-trivial solution of \eqref{pde-k2}, which corresponds to the Minkowski spacetime, reads  
\begin{align*}
F(x,y)&=a(x)b(y), & G(y,z)&=b(y)c(z), & H(x,z)&=c(z)a(x),
\end{align*}
where $a$, $b$ and $c$ are non-zero functions of one real variable. Then
\[
{G_{,z}H-H_{,z}G}={H_{,x}F-F_{,x}H}={F_{,y}G-G_{,y}F}=0
\]
and, consequently, $\kappa=0$.

Thus in the sequel we will assume that $\kappa$ is non-zero. However, we do not claim that in this way we have excluded all solutions to \eqref{pde-k2}, which correspond to the Minkowski spacetime. 

Let us finally present a few solutions to Equation \eqref{pde-k2}.

\paragraph{Example 1.} Suppose that $F$ is constant and that $G$ and $H$ depend merely on the coordinate $z$. Then \eqref{pde-k2} reduces to
\[
G'H-H'G=\pm|F\kappa|,
\]
where the prime denotes the derivative with respect to $z$. Hence we have
\[
G'=\frac{H'}{H}G\pm\frac{|F\kappa|}{H}.
\]
Now let us choose $H$ to be any function of non-zero values and threat the equation above as a non-homogeneous linear differential equation imposed on $G$. A general solution to it is
\begin{equation}
G(z)=\pm|F\kappa|H(z)\int \big(H(z)\big)^{-2}\,dz,
\label{Gz}
\end{equation}
where the integration constant is hidden in the integral.

\paragraph{Example 2.} Let 
\begin{align}
F(x,y)&=a(x)-b(y), & G(y,z)&=b(y)-c(z), & H(x,z)&=c(z)-a(x),
\label{minus}
\end{align}
where $a$, $b$ and $c$ are functions of one real variable. Then \eqref{pde-k2} simplifies to
\[
(c'(z))^2+(a'(x))^2+(b'(y))^2=\kappa^2\neq 0.
\]
This means that each derivative $a'$, $b'$ and $c'$ has to be constant. Consequently, 
\begin{align}
a(x)&=a_1x+a_0, & b(y)&=b_1y+b_0, & c(z)&=c_1z+c_0, 
\label{abc}
\end{align}
where $a_i$, $b_i$ and $c_i$ are constants such that
\begin{equation}
a^2_1+b^2_1+c^2_1=\kappa^2\neq 0.
\label{abc-kappa}
\end{equation}

\paragraph{Example 3.} Suppose now that in Equation \eqref{pde-k2} the functions $G$ and $H$ are constant and non-zero. Then the equation reduces to a two-dimensional one imposed on $F(x,y)$:
\begin{equation}
(F_{,x})^2\frac{H^2}{G^2}+(F_{,y})^2\frac{G^2}{H^2}=\kappa^2.
\label{F2}
\end{equation}
It is convenient to introduce two new coordinate systems $(\bar{x},\bar{y})$ and $(r,\varphi)$ such that
\begin{align}
\bar{x}&:=\kappa\frac{G}{H}x, & \bar{y}&:=\kappa\frac{H}{G}y,\label{bxby}\\
\bar{x}&=:r\cos\varphi, & \bar{y}&=:r\sin\varphi \nonumber
\end{align}
Then Equation \eqref{F2} reads
\begin{equation}
(F_{,\bar{x}})^2+(F_{,\bar{y}})^2=(F_{,{r}})^2+\frac{1}{r^2}(F_{,{\varphi}})^2=1.
\label{F2-1}
\end{equation}
We thus see that $F$ has to satisfy an eikonal equation (see e.g. \cite{eik}).

There are two obvious solutions to the equation above:
\begin{align}
F&=a\bar{x}+b\bar{y}+c, & F&=r+c,
\label{F-obv}
\end{align}
where $a$, $b$ and $c$ are constants such that
\[
a^2+b^2=1.
\]
Another solution to \eqref{F2-1} can be obtained by the following ansatz:
\begin{equation}
F=f(r)+a\varphi,
\label{F-exc}
\end{equation}
where $a$ is a constant. Then \eqref{F2-1} implies that
\[
f(r)=\pm\int\frac{\sqrt{r^2-a^2}}{r}dr=\pm\Big[\sqrt{r^2-a^2}-a\arccos\Big(\frac{a}{r}\Big)\Big]+b,
\]
where $b$ is the integration constant. 

A large class of solutions to the eikonal equation \eqref{F2-1} is presented in \cite{eik} in an implicit form. In this class a single solution $F(\bar{x},\bar{y})$ defined on an open set $U\subset \R^2$ is specified by a choice of \emi a differentiable function $\Psi$ of one real variable and \emii a differentiable function $\tau(\bar{x},\bar{y})$. If on $U$ 
\begin{equation}
\bar{x}-\frac{\tau}{\sqrt{1-\tau^2}}\bar{y}+\Psi'(\tau)=0
\label{tau}
\end{equation}
(where $\Psi'$ is the derivative  of $\Psi$), then
\[
F(\bar{x},\bar{y})=\bar{x}\,\tau+\bar{y}\sqrt{1-\tau^2}+\Psi(\tau).
\] 

This implicit solution can be used to obtain some fairly non-trivial explicit solutions to the eikonal equation. For instance, suppose that $a$ is a non-zero real number and insert 
\begin{equation}
\Psi'(\tau)=a\Big(\frac{\tau}{\sqrt{1-\tau^2}}\Big)^2
\label{Psi'}
\end{equation}
into the condition \eqref{tau} treated now as an equation for an unknown function $\tau$. This equation can be easily solved yielding the following two functions:
\begin{equation}
\tau_{\pm}(\bar{x},\bar{y})=\sgn(a)\frac{\pm\sqrt{\bar{y}^2-4a\bar{x}}+\bar{y}}{\sqrt{4a^2+(\pm\sqrt{\bar{y}^2-4a\bar{x}}+\bar{y})^2}}.
\label{tau-expl}
\end{equation}
On the other hand, it follows from \eqref{Psi'} that
\[
\Psi(\tau)=a\Big(-\tau+\frac{1}{2}\ln\Big|\frac{1+\tau}{1-\tau}\Big|\Big)+b,
\]
where $b$ is the integration constant. Thus 
\[
F(\bar{x},\bar{y})=\bar{x}\tau_{\pm}+\bar{y}\sqrt{1-\tau^2_{\pm}}+a\Big(-\tau_{\pm}+\frac{1}{2}\ln\Big|\frac{1+\tau_{\pm}}{1-\tau_{\pm}}\Big|\Big)+b
\]
with $\tau_{\pm}$ given by \eqref{tau-expl}, is an explicit solution to the eikonal equation \eqref{F2-1}. 

More explicit solutions of this sort will be derived in the forthcoming paper \cite{prep-ao-js}.

\subsection{Summary of the solutions}

All the solutions to the TEGR constraints, we found above, are of the following form:
\begin{align*}
\theta^1&=\pm\frac{G}{FH}dx, & \theta^2&=\pm\frac{H}{GF}dy, & \theta^3&=\pm\frac{F}{HG}dz,\\
r_I&=\frac{\kappa}{\sqrt{3}}\eps_{IKL}\,\theta^K\we\theta^L, & \xi^I&=0, & \zeta_I&=0,
\end{align*}
where the signs $\pm$ can be chosen independently for every $\theta^I$ and $\kappa$ is any non-zero real number. The last equation above follows from \eqref{boost-sol} and the assumption $\beta^{[IJ]}=0$. The solutions are grouped into three families, each family is specified by special form of the functions $F$, $G$ and $H$.

\paragraph{Family 1} Here $F$ is constant, $G$ is given by \eqref{Gz} and $H$ is any function of non-zero values. This gives us the following coframe:
\begin{align*}
\theta^1&= \pm\Big(\kappa\int \big(H(z)\big)^{-2}\,dz\Big)dx,\\
\theta^2&=\pm F^{-2}\Big(\kappa\int \big(H(z)\big)^{-2}\,dz\Big)^{-1}dy,\\  
\theta^3&= \pm \big(H(z)\big)^{-2}\Big(\kappa\int \big(H(z)\big)^{-2}\,dz\Big)^{-1}dz.
\end{align*}
To simplify the form of the solutions let us introduce a new coordinate system $(\bar{x},\bar{y},\bar{z})$:
\begin{align*}
  \bar{x}&:=\pm\kappa x, & \bar{y}&:=\pm\frac{y}{F^2\kappa}, & \bar{z}&:= \ln\Big|\int \big(H(z)\big)^{-2}\,dz\Big|.
\end{align*}  
In these coordinates
\begin{align*}
\theta^1&= e^{\bar{z}}d\bar{x}, &
\theta^2&= e^{-\bar{z}}d\bar{y}, &  
\theta^3&= \pm \frac{1}{\kappa}d\bar{z}.
\end{align*}

\paragraph{Family 2} This family is given by Equations \eqref{minus}  and \eqref{abc}. The solutions here are labeled by six constants $(a_i,b_i,c_i)_{i\in\{0,1\}}$, which satisfy the condition \eqref{abc-kappa}. The corresponding coframe reads:
\begin{equation}
\begin{aligned}
\theta^1&=\pm\frac{b_1y+b_0-c_1z-c_0}{(a_1x+a_0-b_1y-b_0)(c_1z+c_0-a_1x-a_0)}dx,\\
\theta^2&=\pm\frac{c_1z+c_0-a_1x-a_0}{(b_1y+b_0-c_1z-c_0)(a_1x+a_0-b_1y-b_0)}dy,\\
\theta^3&=\pm\frac{a_1x+a_0-b_1y-b_0}{(c_1z+c_0-a_1x-a_0)(b_1y+b_0-c_1z-c_0)}dz.
\end{aligned}
\label{th--}
\end{equation}
Note that if e.g. $a_1\neq 0$, then we can introduce a new coordinate $\bar{x}:=a_1x+a_0$ to simplify the results above.

\paragraph{Family 3} Here $G$ and $H$ are constants, and $F(x,y)$ satisfies the eikonal equation \eqref{F2-1}. It is convenient to express the corresponding coframe in a coordinates system $(\bar{x},\bar{y},\bar{z})$, where $(\bar{x},\bar{y})$ are given by \eqref{bxby} and $\bar{z}:=\kappa z/GH$:
\begin{align*}
\theta^1&=\pm\frac{1}{\kappa F}d\bar{x}, & \theta^2&=\pm\frac{1}{\kappa F}d\bar{y}, & \theta^3&=\pm\frac{F}{\kappa}d\bar{z}.
\end{align*}
For $F$ given by \eqref{F-exc} the corresponding coframe reads 
\begin{align*}
\theta^1&=\pm\frac{1}{\kappa F}(\cos\varphi\,dr-r\sin\varphi\,d\varphi), & \theta^2&=\pm\frac{1}{\kappa F}(\sin\varphi\,dr+r\cos\varphi\,d\varphi), & \theta^3&=\pm\frac{F}{\kappa}d\bar{z}.
\end{align*}

\section{Summary and outlook}

In this paper we considered the Hamiltonian formulation of TEGR introduced in \cite{q-suit,ham-nv} and the constraints on the phase space of the formulation. We showed that the ``position'' variables $(\xi^I)$ can be always gauge-transformed to zero on the whole spatial slice $\Sigma$ of the spacetime. This gauge simplified considerably the constraints in this formulation (compare Equations \eqref{boost-c}--\eqref{scal-c} with \eqref{brv-c} and \eqref{*scal}). Moreover, in this gauge the other ``position'' variables $(\theta^I)$ form a coframe orthonormal w.r.t. the spatial metric $q$.  

We then expressed the fields appearing in the constraints in terms of their components in the coframe $(\theta^I)$ and decomposed $(d\theta^I)$ and $(r_I)$ into the irreducible representations of $SO(3)$. This allowed us to find simple general solutions to the boosts and rotation constraints (Equations \eqref{boost-sol} and \eqref{Ar=0}), and give the scalar constraint the transparent algebraic form \eqref{scal-cmp}.

Seeking to further simplify the constraints, we considered the following issue: is it possible to zero the antisymmetric part of $(d\theta^I)$ by means of the $SO(3)$ gauge transformations generated by the rotation constraint? We made a preliminary analysis of the problem: we derived Equation \eqref{lin}, being a necessary and sufficient condition for an $SO(3)$ gauge transformation, which annihilates the antisymmetric part. We also analyzed the related issue: does a three-dimensional Riemannian metric admit an orthonormal coframe, whose exterior derivative is lacking its antisymmetric part?  This issue has been reduced to the problem of the existence of solutions to Equation \eqref{om2-ABC} being a system of six quadratic first-order PDE's imposed on nine unknown functions.  

We found also three families of exact solutions to all the constraints. To derive these solutions we assumed except the gauge $\xi^I=0$ that \emi $(\theta^I)$ is a Darboux coframe, \emii $(d\theta^I)$ has a vanishing antisymmetric part and \emiii the traceless symmetric part of $(r^I)$ is zero. Two of the three families consist of rather simple solutions, while the third one seems to be more interesting: the solutions in this family are given by solutions to the eikonal equation \eqref{F2-1}, which can be fairly complex.

Further analysis of the TEGR constraints and the solutions will be carried out in \cite{prep-ao-js}. There we will relate the constraints to those known from the standard formulation of General Relativity (called sometime Einstein constraints). Regarding the solutions: we will show  that each solution found in the present paper defines \emi a non-flat spatial metric $q$ on $\Sigma$ and \emii a non-flat spacetime metric $g$ on $\Mc$. We will also present new solutions to the TEGR constraints.

\paragraph{Acknowledgments} I am very grateful to Jerzy Lewandowski, Pawe{\l} Nurowski, \mbox{Piotr} So{\l}tan and Jakub Szymankiewicz for discussions and hints and to Adam Szereszewski for his help with some calculations.

\end{document}